\documentclass{article}
\usepackage{amssymb,amsmath,amsthm,a4wide}
\usepackage[colorlinks=true,linkcolor=black,citecolor=black,runcolor=black]{hyperref}
\usepackage{graphicx}
\usepackage{enumitem}
\usepackage{relsize}
\usepackage{tikz}
\usetikzlibrary{calc}

\usepackage[T2A,T1]{fontenc}
\usepackage[utf8]{inputenc}
\usepackage[russian,english]{babel}

\newtheorem{theorem}{Theorem}[section]
\newtheorem*{lemma*}{Lemma}
\newtheorem*{theorem*}{Theorem}
\newtheorem{proposition}[theorem]{Proposition}
\newtheorem{corollary}[theorem]{Corollary}
\newtheorem*{corollary*}{Corollary}
\newtheorem{lemma}[theorem]{Lemma}

\newtheorem{definition}{Definition}[section]
\newtheorem*{question}{Question}

\theoremstyle{remark}
\newtheorem{remark}{Remark}[section]

\DeclareMathOperator{\KS}{\mathrm{C}\mskip 1.2mu}
\DeclareMathOperator{\KP}{\mathrm{K}\mskip 1.2mu}
\DeclareMathOperator{\mm}{\mathsf{m}\mskip 1.2mu}
\DeclareMathOperator{\ee}{\mathsf{e}\mskip 1.2mu}
\DeclareMathOperator{\E}{\mathrm{E}\mskip 1.2mu}
\DeclareMathOperator{\D}{\mathrm{D}\mskip 1.2mu}
\newcommand{\cnd}{\mskip 1mu|\mskip 1mu}

\newcommand{\mcI}{\mathcal{I}}

\let\le=\leqslant
\let\ge=\geqslant
\let\eps=\varepsilon

\newcommand{\param}{5}

\newenvironment{simpleitem}{\begin{itemize}[leftmargin=*,itemsep=3pt,topsep=3pt,parsep=0pt,partopsep=0pt,label=--]}{\end{itemize}}

\title{Precise Expression for the Algorithmic Information Distance}

\author{Bruno Bauwens\footnote{
  A part of the work was presented in STACS-2020, see~\cite{infoDistRevisitedSTACS} for the abstract.
  \newline
  The work was initiated by Alexander (Sasha) Shen, who informed me about the error in~\cite{mahmud} during a discussion of the paper~\cite{Vitanyi2017}.
  Afterwards I explained the proof of theorem~\ref{th:gap} to Sasha. He simplified it, and he wrote the first 5 sections of the paper.
  Later, I added theorem~\ref{th:exact}, which is proven in section~\ref{sec:exact}.
  After this was added, Sasha decided that his contribution was no longer proportional, and did no longer want to remain coauthor. 
  Finally, I added  theorem~\ref{th:exactSet}.
  I~am especially grateful for Sasha's generous permission to publish the nicely written sections 1--5, 
  with minor modifications suggested by the reviewers of STACS. 
  I thank these reviewers for their suggestions. 
  All errors in this document are solely my responsibility.
  \newline
  \indent I thank Mikhail Andreev for the proof of proposition~\ref{prop:triangle_inequality} and many useful discussions.
  Finally, I thank Artem Grachev and the participants of the Kolmogorov seminar in Moscow state university for useful discussions.
}
}

\begin{document}
\maketitle

\begin{abstract}
  We consider the notion of information distance between two objects $x$ and $y$ 
  introduced by Bennett, G\'acs, Li, Vit\'anyi, and Zurek~\cite{bglvz} 
  as the minimal length of a program that computes $x$ from $y$ as well as computing $y$ from $x$.
  In this paper it was proven that the distance is equal to $\max (\KP(x\cnd y),\KP(y\cnd x))$ up to additive logarithmic terms, 
  and it was conjectured that this could not be improved to $O(1)$ precision.
  We revisit subtle issues in the definition and prove this conjecture. 
  We show that if the distance is at least logarithmic in the length, then this equality does hold with $O(1)$ precision for strings of equal length. 
  Thus for such strings, both the triangle inequality and the characterization hold with optimal precision.
  Finally, we extend the result to sets $S$ of bounded size. 
  We show that for each constant~$s$, the shortest program that prints an $s$-element set $S \subseteq \{0,1\}^n$ given any of its elements, 
  has length at most $\max_{w \in S} \KP(S \cnd w) + O(1)$, 
  provided this maximum is at least logarithmic in~$n$.
\end{abstract}

\section{Introduction}

Informally speaking, Kolmogorov complexity measures the amount of information in an object (say, a bit string) in bits. The complexity $\KS(x)$ of $x$ is defined as the minimal bit length of a program that generates~$x$. This definition depends on the programming language used, but one can fix an optimal language that makes the complexity function minimal up to an $O(1)$ additive term. In a similar way one can define the \emph{conditional} Kolmogorov complexity $\KS(x\cnd y)$ of a string~$x$ given some other string~$y$ as a condition. Namely, we consider the minimal length of a program that transforms~$y$ to~$x$. Informally speaking, $\KS(x\cnd y)$ is the amount of information in~$x$ that is missing in~$y$, the number of bits that we should give in addition to~$y$ if we want to specify $x$.

The notion of \emph{information distance} was introduced in~\cite{bglvz} as ``the length of a shortest binary program that computes $x$ from  $y$ as well as computing $y$ from $x$''. It is clear that such a program cannot be shorter than $\KS(x\cnd y)$  or $\KS(y\cnd x)$ since it performs both tasks; on the other hand, it cannot be much longer than the sum of these two quantities (we can combine the programs that map $x$ to $y$ and vice versa with a small overhead needed to separate the two parts and to distinguish $x$ from $y$). As the authors of~\cite{bglvz} note, ``being shortest, such a program should take advantage of any redundancy between the information required to go from $x$ to $y$ and the information required to go from $y$ to~$x$'', and the natural question arises: to what extent is this possible? The main result of~\cite{bglvz} gives the strongest upper bound possible and says that the information distance equals
$
\max(\KS(x\cnd y),\KS(y\cnd x))
$
with logarithmic precision. In many applications, this characterization turned out to be useful, see~\cite[section 8.4]{lv}.
In fact, in~\cite{bglvz} the prefix version of complexity, denoted by $\KP(x\cnd y)$, and the corresponding definition of information distance were used; see, e.g.~\cite{usv} for the detailed explanation of different complexity definitions. 
The difference between prefix and plain versions is logarithmic in the complexity, 
so it does not matter whether we use plain or prefix versions if we are interested in results with logarithmic precision. 
However, the prefix version of the above characterization has an advantage: 
after adding a large enough constant, this distance satisfies the triangle inequality. 
The plain variant does not have this property, see remark~\ref{rem:trianglePlain} below. 

Several inequalities that are true with logarithmic precision for plain complexity, become true with $O(1)$-precision if prefix complexity is used. So, one could hope that the information distance is equal to $\max \{\KP(x \cnd y), \KP(y \cnd x)\}$ with $O(1)$-precision.
If this is true, then also the original definition satisfies the triangle inequality (after a constant increase). In~\cite[section VII]{bglvz}, this characterization with $O(1)$-precision was conjectured to be false, and in~\cite{mahmud} it was claimed to be true; in~\cite{lzlm} a similar claim is made with reference to~\cite{bglvz}.\footnote{The authors of~\cite{lzlm} define (section 2.2) the function $\E(x,y)$ as the prefix-free non-bipartite version of the information distance (see the discussion below in section~\ref{subsec:four-versions}) and then write: ``the following theorem proved in [4] was a surprise: Theorem 1. $\E(x,y)=\max(\KS(x\cnd y),\KS(y\cnd x))$''. They do not mention that in the paper they cited as [4] (it is~\cite{bglvz} in our list) there is a logarithmic error term; in fact, they do not mention any error terms (though in other statements the constant term is written explicitly). Probably this is a typo, since more general Theorem 2 in~\cite{lzlm} does contain a logarithmic error term.} Unfortunately, the proof in~\cite{mahmud} contains an error, and we show that the result is not valid for prefix complexity with $O(1)$-precision. On the other hand, it is easy to see that the original argument from~\cite{bglvz} can be adapted for plain complexity to obtain the result with $O(1)$-precision, as noted in~\cite{Vitanyi2017}.

In this paper we try to clarify the situation. We discuss the possible definitions of information distance in plain and prefix versions, and their subtle points (one of these subtle points was the source of the error in~\cite{mahmud}). Then we prove our main results, which apply to the prefix distance defined in~\cite{bglvz} and 3 other variants that we discuss in section~\ref{subsec:four-versions}. 

\bigskip
\noindent
Let $ \E(x,y) = \max \{\KP(x \cnd y), \KP(y \cnd x)\}.$  Our first 2 main results are the following.

\newcommand{\theoremExact}{
    If both $x$ and $y$ have length exactly~$n$ and if $\E(x,y) \ge 1.01\log n$, then the prefix information distances are equal to $\E(x,y) + O(1)$. 
  }

\begin{theorem}\label{th:exact}
  \theoremExact
\end{theorem}

\noindent
More generally, for all strings $x$ and $y$ (with possibly different lengths), 
the prefix information distances are equal to $\E(x,y) + O(\log\log \KP(x,y))$, 
see corollary~\ref{cor:exact_loglog} below.
This improves the known precision from logarithmic to double logarithmic in~$\KP(x,y)$. 
In section~\ref{ss:exactStrong} we present more characterizations.

Contrary to what is claimed in~\cite{mahmud}, 
we show that the condition $\E(x,y) \ge 1.01 \log n$ in theorem~\ref{th:exact} is necessary.

\newcommand{\theoremGap}[1]{
   For all $n$, there exist $x$ and $y$ of length $n$ such that all prefix information distances exceed #1 by at least $\Omega(\log \log n)$.
 }

\begin{theorem}\label{th:gap}
  \theoremGap{$\E(x,y)$}
\end{theorem}

\noindent
The difference is bounded by $O(\log \E(x,y))$. The theorem implies that
this is optimal up to constant factors, (thus the characterization does not hold with precision~$O(\log \log \E(x,y))$).
These 2 results above provide 2 surprising precedents regarding the precision of an equality with Kolmogorov complexities.
\\- The plain variant of the characterization is more precise than the prefix variant.
\\- The equality becomes more precise when the quantities become larger.

\bigskip
\noindent
Our last main result generalizes theorem~\ref{th:exact} from pairs of strings to finite sets.
In~\cite{lzlm}, the minimal lenght of a program that maps any element of the set to any other element is studied. 
Such a program must exploit any information shared by all elements of the set.
Given a machine $U$ and a finite set~$S$, let $\D_U(S)$ be the minimal length of a program on~$U$ that on input any element of $S$ 
prints all elements in $S$ and halts.
The distance between strings $x$ and $y$ is the special case of this measure for $S = \{x,y\}$.
In~\cite{Vitanyi2017} it was shown that for optimal plain machines $U$ and for all finite sets $S \subseteq \{0,1\}^*$, 
\[
  \D_U(S) \; = \; \max_{x \in S} \KS(S \cnd x) \;+\; O(\log \# S).
\]
Note that the precision $O(\log \# S)$ does not depend on the length of the strings in~$S$.
It is also shown that this precision is optimal up to constant factors. 
Our last main result provides a similar characterization for the prefix variant.  

\begin{theorem}\label{th:exactSet} 
  If $U$ is a prefix-free machine that makes the function $\D_U$ minimal up to additive constants, then 
  \[
    \D_U(S) \;=\; \max_{x \in S} \KP(S \cnd x) \;+\; O(\log \#S),
  \]
  provided $S \subseteq \{0,1\}^n$ and the maximum is at least~$(\param\#S )^{\# S}\log n$.
\end{theorem}

\noindent
This implies that for sets $S \subseteq \{0,1\}^*$ of any fixed size: $\D_U(S) = \max \KP(S \cnd x) + O(\log \log K(S))$.
We also provide a different and incomparible condition for the equality of the theorem.
It holds for all sets $S \subseteq \{0,1\}^n$ in which all different elements $u$ and $v$ satisfy $\E(u,v) \ge \param\#S \log (n\# S)$, see  proposition~\ref{prop:exactSet}.

\bigskip
\noindent
The theorems are proven using the game technique, which means that we present a 2-person game, and obtain the result from a winning strategy for one of the players,
see~\cite{BauwensCompcomp,ShenCie12,VerSurvey} for other examples. 
Our strategy uses ideas from~\cite{gacs1983}. 
In~\cite{nidTheoretical} the normalized version of the information distance was studied, which has values in the interval $[0,1]$ when defined with a suitable optimal machine.
In~\cite{BauwensNID}, the game technique was used to prove that no semicomputable function differs from this normalized distance by less than~$0.5$.

We discuss the plain information distance in section~\ref{sec:plain}. Then, in section~\ref{sec:prefix} we discuss the different definitions of prefix complexity (with prefix-free and prefix-stable machines, as well as definitions using the a priori probability), and in section~\ref{sec:prefix-distance} we discuss their counterparts for the information distance. In sections~\ref{sec:game} and~\ref{sec:exact} we prove the first 2 main results. 
In section~\ref{sec:nonsharedSet} we review generalizations for sets and prove the third main result. Finally, we present open questions in section~\ref{sec:openQuestions}. 

\section{Plain complexity and information distance}\label{sec:plain}

Let us recall the definition of plain conditional Kolmogorov complexity. Let $U(p,x)$ be a computable partial function of two string arguments; its values are also binary strings. We may think of $U$ as an interpreter of some programming language. The first argument $p$ is considered as a program and the second argument is an input for this program. Then we define the complexity function
\[
\KS_U (y\cnd x) \;=\; \min\{|p|\colon U(p,x) \mathop = y\};
\]
here $|p|$ stands for the length of a binary string $p$, so the right hand side is the minimal length of a program that produces output $y$ given input~$x$. The classical Solomonoff--Kolmogorov theorem says that there exists an optimal $U$ that makes $\KS_U$ minimal up to an $O(1)$-additive term. We fix some optimal $U$ and then denote $\KS_U$ by just $\KS$. See, e.g., \cite{lv,usv} for the details.

Now we want to define the information distance between $x$ and $y$. One can try the following approach: take some optimal $U$ from the definition of conditional complexity and then define
\[
 \D_U(x,y) \;=\; \min\{|p|\colon U(p,x) \mathop=y \text{ and } U(p,y) \mathop=x\},
 \]
i.e., consider the minimal length of a program that both maps $x$ to $y$ and $y$ to $x$. However, there is a caveat, as the following simple observation shows.

\begin{proposition}
There exists some computable partial function $U$ that makes $\KS_U$ minimal up to an $O(1)$ additive term, and still $\D_U(x,y)$ is infinite for some strings $x$ and $y$ and therefore not minimal.
\end{proposition}

\begin{proof}
Consider an optimal function $U$ and then define $U'$ such that $U(\Lambda,x)=\Lambda$  where $\Lambda$ is the empty string, $U'(0p,x)=0U(p,x)$ and $U'(1p,x)=1U(p,x)$. In other terms, $U'$ copies the first bit of the program to the output and then applies $U$ to the rest of the program and the input. It is easy to see that $\KS_{U'}$ is minimal up to an $O(1)$ additive term, but $U'(q,\cdot)$ has the same first bit as $q$, so if $x$ and $y$ have different first bits, there is no $q$ such that $U'(q,x)=y$ and $U'(q,y)=x$ at the same time.
\end{proof}
\noindent
On the other hand, the following proposition is true (and can be proven in the same way as the existence of the optimal $U$ for conditional complexity):

\begin{proposition}\label{prop:defPlainID}
There exists a computable partial function $U$ that makes $\D_U$ minimal up to $O(1)$ additive term.
\end{proposition}

\noindent
Now we may define the plain information distance as the minimal function~$\D_U$. 
For example, $\D_U(x,\text{empty string}) = \KS(x) + O(1)$, by considering a program for $x$ and modify it such that on input the empty string prints $x$ and otherwise it prints the empty string. 
For all $n$-bit $x$ and $y$, we have $\D_U(x,y) \le n + O(1)$, because knowing the bitwise XOR of $x$ and $y$, we can map $x$ to $y$ and vice versa.

\bigskip
\noindent
It turns out that the original argument from~\cite{bglvz} can be easily adapted to show the following result 
(that is a special case of a more general result about several strings proven in~\cite{Vitanyi2017}):

\begin{theorem}\label{th:bglvz}
The minimal function $\D_U$ equals $\max(\KS(x\cnd y),\KS(y\cnd x))+O(1)$.
\end{theorem}

\begin{proof}
  We provide the adapted proof for later reference. 
  In one direction we have to prove that $\KS(x\cnd y)\le \D_U(x,y)+O(1)$, and the same for $\KS(y\cnd x)$. This is obvious, since the definition of $\D_U$ contains more requirements for $p$, (it should map both $x$ to $y$ and $y$ to $x$, while in $\KS(x\cnd y)$ it is enough to map $y$ to $x$).

To prove the reverse inequality, consider for each $n$ the binary relation $R_n$ on strings (of all lengths) defined as
$$
R_n(x,y) \quad \Longleftrightarrow \quad \KS(x\cnd y)< n \text{ and } \KS(y\cnd x) <n.
$$
By definition, this relation is symmetric. It is easy to see that $R_n$ is (computably) enumerable uniformly in $n$, since we may compute better and better upper bounds for $\KS$ reaching ultimately its true value. We think of $R_n$ as the set of edges of an undirected graph whose vertices are binary strings. Note that each vertex $x$ of this graph has degree less than $2^n$ since there are less than $2^n$ programs of length less than $n$ that map $x$ to its neighbors.

For each $n$, we enumerate edges of this graph (i.e., pairs in $R_n$). We want to assign colors to the edges of $R_n$ in such a way that edges that have a common endpoint have different colors. In other terms, we require that for every vertex $x$, all edges of $R_n$ adjacent to $x$ have different colors. For that, $2^{n+1}$ colors are enough. Indeed, each new edge needs a color that differentiates it from less than $2^n$ existing edges adjacent to one its endpoint and less than $2^n$ edges adjacent to other endpoint.

Let us agree to use $(n+1)$-bit strings as colors for edges in $R_n$, and perform  this coloring in parallel for all $n$. Now we define $U(p,x)$ for a $(n+1)$-bit string $p$ and arbitrary string $x$ as the string $y$ such that the edge $(x,y)$ has color $p$ in the coloring of edges from $R_n$. Note that $n$ can be reconstructed as $|p|-1$. The uniqueness property for colors guarantees that there is at most one $y$ such that $(x,y)$ has color $p$, so $U(p,x)$ is well defined. It is easy to see now that if $\KS(x\cnd y)<n$ and $\KS(y\cnd x)<n$, and $p$ is the color of the edge $(x,y)$, then $U(p,x)=y$ and $U(p,y)=x$ at the same time. This implies the reverse inequality (the $O(1)$ terms appears when we compare our $U$ with the optimal one).
\end{proof}

\begin{remark}
In the definition of information distance given above we look for a program $p$ that transforms $x$ to $y$ and also transforms $y$ to $x$. Note that we \emph{do not tell the program which of the two transformations is requested}. A weaker definition would provide also this information to $p$. This modification can be done in several ways. For example, we may require in the definition of $\D$ that $U(p,0x)=y$ and $U(p,1y)=x$, using the first input bit as the direction flag. An equivalent approach is to use two computable functions $U$ and $U'$ in the definition and require that $U(p,x)=y$ and $U'(p,y)=x$. This corresponds to using different interpreters for both directions.

  It is easy to show that the optimal functions $U$ and $U'$ exist for this two-interpreter version of the definition. A priori we may get a smaller value of information distance in this way, because the program's task is easier when the direction is known, informally speaking.  But it is not the case for the following simple reason. Obviously, this new quantity is still an upper bound for both conditional complexities $\KS(x\cnd y)$ and $\KS(y\cnd x)$ with $O(1)$ precision. Therefore theorem~\ref{th:bglvz} guarantees that this new definition of information distance coincides with the old one up to $O(1)$ additive terms. For the prefix versions of information distance such a simple argument does not work anymore, because the variant of theorem~\ref{th:bglvz} for prefix complexity does not hold, see theorem~\ref{th:gap}. 
\end{remark}

\bigskip
\noindent
We have seen that different approaches lead to the same notion of plain information distance (up to $O(1)$ additive term). There is also a simple and natural quantitative characterization of this notion as a minimal function in a class of functions.

\begin{theorem}\label{th:plainDistChar}
  Consider the class of functions $E$ that are 
  symmetric, upper semicomputable, and for some $c$, all $n$ and all $x$, satisfy
  \begin{equation*}\tag{\ensuremath{*}}\label{eq:balls}
  \#\{y\colon E(x,y)<n\} \;\le\; c2^n. 
  \end{equation*}
  For every optimal $U$ this class contains $\D_U$, and for any $E$ in this class, we have $\D_U \le E + O(1)$.
\end{theorem}

\noindent
Recall that upper semicomputability of $E$ means that one can compute a sequence of total upper bounds for $E$ that converges to $E$. The equivalent requirement: the set of triples $(x,y,n)$ where $x,y$ are strings and $n$ are natural numbers, such that $E(x,y)<n$, is (computably) enumerable.

\begin{proof}
  The function $\max(\KS(x\cnd y),\KS(y\cnd x))$ is upper semicomputable and symmetric. The inequality \eqref{eq:balls} is true for it since it is true for the smaller function $\KS(y\cnd x)$ (for $c=1$; indeed, the number of programs of length less than $n$ is at most $2^n$).

  On the other hand, if $E$ is some symmetric upper semicomputable function that satisfies \eqref{eq:balls}, then one can for any given $x$ and $n$ enumerate all $y$ such that $E(x,y)<n$. There are less than $c2^n$ strings $y$ with this property, so given $x$, each such $y$ can be described by a string of $n+\lceil \log c\rceil$ bits, its ordinal number in the enumeration. Note that the value of $n$ can be reconstructed from this string by decreasing its length by $\lceil \log c\rceil$, so $\KS(y\cnd x)\le n+O(1)$ if $E(x,y)<n$. It remains to apply the symmetry of $E$ and theorem~\ref{th:bglvz}.
\end{proof}

\begin{remark}\label{rem:trianglePlain}
The name ``information distance'' motivates the following question: does the plain information distance satisfy the triangle inequality? With logarithmic precision the answer is positive, because 
  \[
\KS(x\cnd z)\;\le\; \KS(x\cnd y)+\KS (y\cnd z)+O(\log (\KS(x\cnd y)+\KS(y\cnd z))).
\]
However, if we replace the last term by an $O(1)$-term, then this inequality is not true. Indeed, for all strings $x$ and $y$, the distance between the empty string $\Lambda$ and $x$ is $\KS(x)+O(1)$, and the distance between $x$ and some encoding of a pair $(x,y)$ is at most $\KS(y)+O(1)$, and the triangle inequality for distances with $O(1)$-precision would imply $\KS(x,y)\le\KS(x)+\KS(y)+O(1)$. But this is not true, see, e.g., \cite[section 2.1]{usv}.
\end{remark}

One may ask whether a weaker statement saying that there is a maximal (up to an $O(1)$ additive term) function in the class of functions 
that both satisfy the conditions of theorem~\ref{th:plainDistChar} and the triangle inequality. 
The answer is negative, as the following proposition shows.

\begin{proposition}\label{prop:triangle_inequality}
  There are two upper semicomputable symmetric functions $E_1$, $E_2$ that both satisfy the condition \eqref{eq:balls} and the triangle inequality, such that no function that is bounded both by $E_1$ and $E_2$ can satisfy \eqref{eq:balls} and the triangle inequality at the same time.
\end{proposition}

\begin{proof}
  Let us agree that $E_1(x,y)$ and $E_2(x,y)$ are infinite when $x$ and $y$ have different lengths. If $x$ and $y$ are $n$-bit strings, then $E_1(x,y)\le k$ means that all the bits in $x$ and $y$ outside the first $k$ positions are the same, and $E_2(x,y)\le k$ is defined in a symmetric way for the last $k$ positions. Both $E_1$ and $E_2$ satisfy the triangle inequality (and even the ultrametric inequality) and also satisfy condition \eqref{eq:balls}, since the ball of radius $k$ consist of strings that coincide except for the first/last $k$ bits. If $E$ is bounded both by $E_1+O(1)$ and $E_2+O(1)$ and satisfies the triangle inequality, then by changing the first $k$ and the last $l$ positions in a string $x$ we get a string $y$ such that $E(x,y)\le k+l+ O(1)$. It is easy to see that the number of strings $y$ that can be obtained in this way for all $k$ and $\ell$ with $k + \ell = n/2$ is not $O(2^{n/2})$, but $\Theta(n2^{n/2})$.
\end{proof}

\section{Prefix complexity: different definitions}\label{sec:prefix}

The notion of prefix complexity was introduced independently by Levin~\cite{gacs1974,levin1971,levin1974} and later by Chaitin~\cite{chaitin1975}. There are several versions of this definition, and they all turn out to be equivalent, so people usually do not care much about technical details that are different. However, if we want to consider the counterparts of these definitions for information distance, their differences become important if we are interested in $O(1)$-precision. 

Essentially there are four different definitions of prefix complexity that appear in the literature.

\subsection{Prefix-free definition}

A computable partial function $U(p,x)$ with two string arguments and string values is called \emph{prefix-free} (with respect to the first argument) if $U(p,x)$ and $U(p',x)$ cannot be defined simultaneously for a string $p$ and its prefix $p'$ \emph{and for the same second argument $x$}. In other words, for every string $x$ the set of strings $p$ such that $U(p,x)$ is defined is prefix-free, i.e., does not contain a string and its prefix at the same time.

For a prefix-free function $U$ we may consider the complexity function $\KS_U(y\cnd x)$. In this way we get a smaller class of complexity functions compared with the definition of plain complexity, 
and the Solomonoff--Kolmogorov theorem can be easily modified to show that there exists a minimal complexity function in this smaller class (up to $O(1)$ additive term, as usual). This function is called \emph{prefix conditional complexity} and usually is denoted by $\KP(y\cnd x)$. It is greater than $\KS(y\cnd x)$ since the class of available functions $U$ is more restricted; the relation between $\KS$ and $\KP$ is well studied, see, e.g.,~\cite[chapter 4]{usv} and references within.

The unconditional prefix complexity $\KP(x)$ is defined in the same way, with $U$ that does not have a second argument. We can also define $\KP(x)$ as $\KP(x\cnd y_0)$ for some fixed string $y_0$. This string may be chosen arbitrarily; for each choice we have $\KP(x)=\KP(x\cnd y_0)+O(1)$ but the constant in the $O(1)$ bound depends on the choice of $y_0$.

\subsection{Prefix-stable definition}

The prefix-stable version of the definition considers another restriction on the function~$U$. Namely, in this version the function $U$ should be \emph{prefix-stable} with respect to the first argument. This means that if $U(p,x)$ is defined, then $U(p',x)$ is defined and equal to $U(p,x)$ for all $p'$ that are extensions of $p$ (i.e., when $p$ is a prefix of $p'$). We consider the class of all computable partial prefix-stable functions $U$ and corresponding functions $\KS_U$, and observe that there exists an optimal prefix-stable function $U$ that makes $\KS_U$ minimal in this class. 

It is rather easy to see that the prefix-stable definition leads to a version of complexity that is bounded by the prefix-free one, since each prefix-free computable function can be easily extended to a prefix-stable one. The reverse inequality is not so obvious and there is no known direct proof; the standard argument compares both versions with the forth definition of prefix complexity, (the logarithm of a maximal semimeasure, see section~\ref{subsec:semimeasure} below). 

Prefix-free and prefix-stable definitions correspond to the same intuitive idea: the program should be ``self-delimiting''. This means that the machine gets access to an infinite sequence of bits that starts with the program and has no marker indicating the end of a program. The prefix-free and prefix-stable definitions correspond to two possible ways of accessing this sequence. The prefix-free definition corresponds to a blocking read primitive, which means that if the program is given to the machine bit by bit in a queue, then upon each request of the next bit, the computation halts until the bit is provided. In this setting, the program itself decides when it has seen enough bits. The prefix-stable definition corresponds to a non-blocking read primitive, which means that if the bit is not provided, the machine may continue computations and may still produce an output, even if the requested bit is never provided, but if the bit were provided, then the output should be the same. For more details we refer \mbox{to~\cite[section~4.4]{usv}}.

\subsection{A priori probability definition}\label{subsec:apriori}

In this approach we consider the \emph{a priori probability} of $y$ given $x$, the probability of the event ``a random program maps $x$ to $y$''. More precisely, consider a prefix-stable function $U(p,x)$ and an infinite sequence $\pi$ of independent uniformly distributed random bits (a random variable). We say that $U(\pi,x)=y$ if $U(p,x)=y$ for some $p$ that is a prefix of $\pi$. Since $U$ is prefix-stable, the value $U(\pi,x)$ is well defined. For given $x$ and $y$, we denote by $m_U(y\cnd x)$ the probability of this event (the measure of the set of $\pi$ such that $U(\pi,x)=y$). For each prefix-stable $U$ we get some function $m_U$. It is easy to see that there exists an optimal $U$ that makes $m_U$ maximal (up to an $O(1)$-factor). Then we define prefix complexity $\KP(y\cnd x)$ as $-\log m_U(y\cnd x)$ for this optimal $U$, where the logarithm has base 2.

It is also easy to see that if we use prefix-free functions $U$ instead of prefix-stable ones, we obtain the same definition of prefix complexity. Informally speaking, if we have an infinite sequence of random bits as the first argument, we do not care whether we have blocking or non-blocking read access, the bits are always there. The non-trivial and most fundamental result about prefix complexity is that this definition, as the logarithm of the probability, is equivalent to the two previous ones. As a byproduct of this result we see that the prefix-free and prefix-stable definitions are equivalent. This proof and the detailed discussion of the difference between the definitions can be found, e.g., in~\cite[chapter 4]{usv}.

\subsection{Semimeasure definition}\label{subsec:semimeasure}

The semimeasure approach defines a priori probability in a different way, as a convergent series that converges as slow as possible. More precisely, a \emph{lower semicomputable semimeasure} is a non-negative real-valued function $m(x)$ on binary strings such that $m(x)$ is a limit of an increasing sequence of rational numbers and $\sum_x m(x)\le 1$ that is computable uniformly in~$x$. There exists a lower semicomputable semimeasure $\mm(x)$ that is maximal up to $O(1)$-factors, and its negative logarithm coincides with unconditional prefix complexity $\KP(x)$ up to an $O(1)$ additive term. 

We can define conditional prefix complexity in the same way, considering semimeasures with parameter $y$. Namely, we consider lower semicomputable non-negative real-valued functions $m(x,y)$ such that  $\sum_x m(x,y)\le 1$ for every $y$. Again there exists a maximal function among them, denoted by $\mm(x\cnd y)$, and its negative logarithm equals $\KP(x\cnd y)$ up to an $O(1)$ additive term.

To prove this equality, we note first that the a priori conditional probability $m_U(x\cnd y)$ is a lower semicomputable conditional semimeasure. The lower semicomputability is easy to see: we can simulate the machine $U$ and discover more and more programs that map $y$ to $x$. The inequality $\sum_x m_U(x\cnd y)$ also has a simple probabilistic meaning: the events ``$\pi$ maps $y$ to $x$'' for a given $y$ and different $x$ are disjoint, so the sum of their probabilities does not exceed $1$. The other direction (starting from a semimeasure, construct a machine) is a bit more difficult, but in fact it is possible (even exactly, without additional $O(1)$-factors). See~\cite[chapter 4]{usv} for details.

The semimeasure definition can be reformulated in terms of complexities by taking exponents: $\KP(x\cnd y)$ is a minimal (up to $O(1)$ additive term) upper semicomputable non-negative integer function $k(x,y)$ such that
\[
\sum_x2^{-k(x,y)} \;\le\; 1
\]
for all $y$.
A similar characterization of plain complexity would use a weaker requirement
$$
\#\{x\colon k(x,y) \mathop < n\} \;<\; c2^n
$$
for some $c$ and all $y$. (We discussed a similar result for information distance where the additional symmetry requirement was used, but the proof is the same.)

\subsection{Warning}\label{subsec:warning}

There exists a definition of plain conditional complexity that does \emph{not} have a prefix-version counterpart. Namely, the plain conditional complexity $\KS(x\cnd y)$ can be equivalently defined as the \emph{minimal unconditional plain complexity of a program that maps $y$ to $x$}.  In this way we do not need the programming language used to map $y$ to $x$ to be optimal; it is enough to assume that we can computably translate programs in other languages into our language; this property, sometimes called \emph{$s$-$m$-$n$-theorem} or \emph{G\"odel property of a computable numbering}, is true for almost all reasonable programming languages. Of course, we still assume that the language used in the definition of unconditional Kolmogorov complexity is optimal. 

One may hope that $\KP(x\cnd y)$ can be similarly defined as the minimal unconditional prefix complexity of a program that maps $y$ to $x$.  The following proposition shows that it is not the case.

\begin{proposition}\label{prop:not-program-complexity}
The prefix complexity $\KP(x\cnd y)$ does not exceed the minimal prefix complexity of a program that maps $y$ to $x$; however, the difference between these two quantities is not bounded.
\end{proposition}

\begin{proof} To prove the first part, assume that $U_1(p)$ is a prefix-stable function of one argument that makes the complexity function 
  \[
\KS_{U_1}(q) \;= \;\min\{ |p|\colon U(p)\mathop =q\}
\]
 minimal. Then $\KS_U(q)=\KP(q)+O(1)$. (We still need an $O(1)$ term since the choice of an optimal prefix-stable function is arbitrary). Then consider the function
  \[
U_2(p,x)\;=\;[U_1(p)](x)
\]
where $[q](x)$ denotes the output of a program $q$ on input $x$. Then $U_2$ is a prefix-stable function from the definition of conditional prefix complexity, and 
  \[
\KS_{U_2}(y\cnd x) \;\le\; \KS_{U_1}(q)
  \]
for any program $q$ that maps $x$ to $y$ (i.e., $[q](x)=y$). This gives the inequality mentioned in the proposition. Now we have to show that this inequality is not an equality with $O(1)$-precision.

Note that $\KP(x\cnd n)\le n+O(1)$ for every binary string $x$ of length $n$. Indeed, a prefix-stable (or prefix-free) machine that gets $n$ as input can copy 
$n$ first bits of its program to the output. (The prefix-free machine should check that there are exactly $n$ input bits.) In this way we get $n$-bit programs for all strings of length $n$. 

Now assume that the two quantities coincide up to an $O(1$) additive term. Then for every string $x$ there exists a program $q_x$ that maps $|x|$ to $x$ and $\KP(q_x)\le |x|+c$ for all $x$ and some~$c$. Note that $q_x$ may be equal to $q_y$ for $x\ne y$, but this may happen only if $x$ and $y$ have different lengths. Consider now the set $Q$ of all $q_x$ for all strings $x$, and the series 
  \[
\sum_{q\in Q} 2^{-\KP(q)}.\eqno(**)
\]
This sum does not exceed $1$ (it is a part of a similar sum for all $q$ that is at most $1$, see above). On the other hand, we have at least $2^n$ different programs $q_x$ for all $n$-bit strings $x$, and they correspond to different terms in $(**)$; each of these terms is at least $2^{-n-c}$. We get a converging series that contains, for every $n$, at least $2^n$ terms of size at least $2^{-n-c}$. It is easy to see that such a series does not exist. Indeed, each tail of this series should be at least $2^{-c-1}$ (consider these $2^n$ terms for large $n$ when at least half of these terms are in the tail), and this is incompatible with convergence. 
\end{proof}

Why do we get a bigger quantity when considering the prefix complexity of a program that maps $y$ to $x$? The reason is that the prefix-freeness (or prefix-stability) requirement for the function $U(p,x)$ is formulated separately for each $x$: the decision where to stop reading the program $p$ \emph{may depend on its input}~$x$. This is not possible for a prefix-free description of a program that maps $x$ to~$y$. It is easy to overlook this problem when we informally describe prefix complexity $\KP(x\cnd y)$ as ``the minimal length of a program, written in a self-delimiting language, that maps $y$ to $x$'', because the words ``self-delimiting language''  implicitly assume that we can determine where the program ends while reading the program text (and before we know its input), and this is a wrong assumption. 

\subsection{Historical digression}\label{subsec:history}

Let us comment a bit on the history of prefix complexity. It appeared first in 1971 in Levin's PhD thesis~\cite{levin1971}; Kolmogorov was his thesis advisor. Levin used essentially the semimeasure definition (formulated a bit differently). This thesis was in Russian and remained unpublished for a very long time. In 1974 G\'acs' paper~\cite{gacs1974} appeared where the formula for the prefix complexity of a pair was proven. This paper mentioned prefix complexity as ``introduced by Levin in [4], [5]'' (\cite{levin1973} and \cite{levin1974} in our numbering).  The first of these two papers does not say anything about prefix complexity explicitly, but defines the monotone complexity of sequences of natural numbers, and prefix complexity can be considered as a special case when the sequence has length $1$ (this is equivalent to the prefix-stable definition of prefix complexity).  The second paper has a comment ``(to appear)'' in G\'acs' paper. We discuss it later in this subsection. 

G\'acs does not reproduce the definition of prefix complexity, saying only that it is ``defined as the complexity of specifying $x$ on a machine on which it is impossible to indicate the endpoint\footnote{The English translation says ``halting'' instead of ``endpoint'' but this is an obvious translation error.} of a master program: an infinite sequence of binary symbols enters the machine and the machine must itself decide how many binary symbols are required for its computation''. This description is not completely clear, but it looks more like a prefix-free definition if we understand it in such a way that the program is written on a one-directional tape and the machine decides where to stop reading.  G\'acs also notes that prefix complexity (he denotes it by $KP(x)$) ``is equal to the [negative] base two logarithm of a universal semicomputable probability measure that can be defined on the countable set of all words''. 

Levin's 1974 paper~\cite{levin1974} says that ``the quantity $KP(x)$ has been investigated in details in [6,7]''. Here [7] in Levin's numbering is G\'acs paper cited above (\cite{gacs1974} in our numbering) and has the comment ``in press'', and [6] in Levin's numbering is cited as \begin{otherlanguage*}{russian}``\emph{Левин Л.А.}, О различных видах алгоритмической сложности конечных объектов (в печати)'' \end{otherlanguage*}
[Levin L.A., On a different version of algorithmic complexity of finite objects, to appear]. Levin does not have a paper with exactly this title, but the closest approximation is his 1976 paper~\cite{levin1976}, where prefix complexity is defined as the logarithm of a maximal semimeasure.  Except for these references, \cite{levin1974} describes the prefix complexity in terms of prefix-stable functions: ``It differs from the Kolmogorov complexity measure $\langle\ldots\rangle$ in that the decoding algorithm $A$ has the following ``prefix'' attribute: if $A(p_1)$ and $A(p_2)$ are defined and distinct, then $p_1$ cannot be a beginning fragment of~$p_2$''.  

The prefix-free and a priori probability definitions were given independently by Chaitin in~\cite{chaitin1975} (in different notation) together with the proof of their equivalence, so~\cite{chaitin1975} was the first publication containing this (important) proof.
 
Now it seems that the most popular definition of prefix complexity is the prefix-free one, for example, it is given as the main definition in~\cite{lv}.

\section{Prefix complexity and information distance}\label{sec:prefix-distance}

\subsection{Four versions of prefix information distance}\label{subsec:four-versions}

Both the prefix-free and prefix-stable versions of prefix complexity have their counterparts for the information distance.

Let $U(p,x)$ be a partial computable prefix-free [respectively, prefix-stable] function of two string arguments having string values. Consider the function 
\[
\D_U(x,y) \;=\; \min\{ |p|\colon U(p,x) \mathop = y \text{ and } U(p,y) \mathop=x\}.
\]
As before, one can easily prove that there exists a minimal (up to $O(1)$) function among all functions $\D_U$ of the class considered.  It will be called the \emph{prefix-free} [respectively~\emph{prefix-stable}] \emph{information distance}. We clarify the difference between these variants.

Note that only the cases when $U(p,x)=y$ and also $U(p,y)=x$ matter for $\D_U$. So we may assume without loss of generality that $U(p,x)=y \Leftrightarrow U(p,y)=x$ waiting until both equalities are true before finalizing the values of $U$. Then for every $p$ we have some matching $M_p$ on the set of all strings: an edge $x$--$y$ is in $M_p$ if $U(p,x)=y$ and $U(p,y)=x$. This is indeed a matching: for every $x$ only $U(p,x)$ may be connected with $x$.

The set $M_p$ is enumerable uniformly in $p$. In the prefix-free version the matchings $M_p$ and $M_q$ are disjoint (have no common vertices) for two compatible strings $p$ and $q$ (one is an extension of the other). For the prefix-stable version $M_p$ increases when $p$ increases (and remains a matching).  It is easy to see that a family $M_p$ that has these properties, always corresponds to some function~$U$, and this statement holds both in the prefix-free and prefix-stable version.

There is another way in which this definition could be modified. As we have discussed for plain complexity, we may consider two different functions $U$ and $U'$ and consider the distance function
\[
\D_{U,U'}(x,y) \;=\; \min\{ |p|\colon U(p,x) \mathop = y \text{ and } U'(p,y) \mathop=x\}.
\]
Intuitively this means that we know the transformation direction in addition to the input string. This corresponds to matchings in a bipartite graph where both parts consist of all binary strings; the edge $x$--$y$ is in the matching $M_p$ if $U(p,x)=y$ and $U'(p,y)=x$.  Again instead of the pair $(U,U')$ we may consider the family of matchings that are disjoint (for compatible $p$, in the prefix-free version) or monotone (for the prefix-stable version). In this way we get two other versions of information distance that could be called \emph{bipartite prefix-free} and \emph{bipartite prefix-stable} information distances.

In~\cite{bglvz} the information distance is defined as the prefix-free information distance with the same function $U$ for both directions, not two different ones. The definition in section III considers the minimal function among all $\D_U$. This minimal function is denoted by $E_0(x,y)$, while $\max(\KP(x\cnd y),\KP(y\cnd x))$ is denoted by $E_1(x,y)$, see section I of the same paper. The inequality $E_1\le E_0$ is obvious, and the reverse inequality with logarithmic precision is proven in~\cite{bglvz} as Theorem~3.3. 

Which of the four versions of prefix information distance is the most natural?  Are they really different? It is easy to see that the prefix-stable version (bipartite or not) is bounded by the corresponding prefix-free version, since every prefix-free function has a prefix-stable extension. Also each bipartite version (prefix-free or prefix-stable) does not exceed the corresponding non-bipartite version for obvious reasons: one may take $U=U'$. It is hard to say which version is most natural, and the question whether some of them coincide or all four are different, remains open. Let $\E(x,y)$ denote the maximum of the conditional prefix complexities. Since the non-bipartite prefix-free distance is the maximal of all 4, the result from~\cite{bglvz} implies the following.

\begin{theorem}\label{th:equality_logarithmic}
  All prefix information distances are equal to $\E(x,y) + O(\log \E(x,y))$.
\end{theorem}
\noindent
This result also follows from theorem~\ref{th:bglvz}. 
Indeed, one can convert a program on a plain machine to a program on a prefix-free machine by prepending a prefix-free description of its length. 
Consider a minimal program for the plain distance, and prepend a prefix-free description of length at most $O(\log \E(x,y))$.
This is possible because by theorem~\ref{th:bglvz} the plain distance is bounded by $\max{\KS(x \cnd y),\KS(y \cnd x)} \le \E(x,y)$ up to $O(1)$ constants.
The length of the concatenation satisfies the bound of theorem~\ref{th:equality_logarithmic}.

As we prove in theorem~\ref{th:gap}, the smallest of all four distances, the prefix-stable bipartite version, is still bigger than the maximum $\E$ of conditional complexities, and the difference is unbounded. Hence, for all four versions, including the prefix-free non-bipartite version used both in~\cite{bglvz,lzlm,mahmud}, the equality with $O(1)$-precision is not true. This confirms the conjecture in section VII of~\cite{bglvz} and contradicts what is claimed in~\cite[Theorem 3.10]{mahmud}. However, if $\E(x,y)$ is at least logarithmic, then all 4 distances are equal to $\E$ with $O(1)$ precision, see theorem~\ref{th:exact}.

Before proving these results, we prove some positive results about the definition of information distance that is a counterpart of the a priori probability definition of prefix complexity.

\subsection{A priori probability of going back and forth}\label{subsec:apriori-distance}

Fix some prefix-free function $U(p,x)$. The conditional a priori probability $m_U (y\cnd x)$ is defined as
$$
\Pr_{\pi} [U(\pi,x)=y],
$$
where $\pi$ is a random infinite sequence, and $U(\pi,x)=y$ means that $U(p,x)=y$ for some $p$ that is a prefix of $\pi$. 
As we discussed, there exists a maximal function among all $m_U$, and its negative logarithm equals the conditional prefix complexity $\KP(y\cnd x)$.

Now let us consider the counterpart of this construction for the information distance. The natural way to do this is to consider the function
$$
e_{U}(x,y)\;=\;\Pr_{\pi} [U(\pi,x) \mathop{=} y \text{ and } U(\pi,y) \mathop{=}x].
$$
Note that in this definition the prefixes of $\pi$ used for both computations are not necessarily the same. It is easy to show, as usual, that there exists an \emph{optimal} machine $U$ that makes $e_U$ maximal. Fixing some optimal $U$, we get some function $\ee(x,y)$. Note that different optimal $U$ lead to functions that differ only by $O(1)$-factor. The negative logarithm of this function coincides with the maximum of the conditional complexities, 
as the following result says.

\begin{theorem}\label{thm:apriori}
  \[
-\log \ee(x,y) \;= \;\max(\KP(x\cnd y),\KP(y\cnd x))+O(1).
\]
\end{theorem}

\begin{proof}
Rewriting the right-hand side in the exponential scale, we need to prove that 
  \[
\ee(x,y)\;=\;\min (\mm(x\cnd y),\mm(y\cnd x))
\]
up to $O(1)$-factors. One direction is obvious: $\ee(x,y)$ is smaller than $\mm(x\cnd y)$ since the set of $\pi$ in the definition of $\ee$ is a subset of the corresponding set for $\mm$, if we use the probabilistic definition of $\mm=m_U$. The same is true for $\mm(y\cnd x)$. 

The non-trivial part of the statement is the reverse inequality. Here we need to construct a machine $U$ such that 
  \[
e_U(x,y) \;\ge\; \min (\mm(x\cnd y),\mm(y\cnd x))
\]
up to $O(1)$-factors.

Let us denote the right-hand side by $u(x,y)$. The function $u$ is symmetric, lower semicomputable and $\sum_y u(x,y)\le 1$ for all $x$ (due to the symmetry, we do not need the other inequality where $y$ is fixed). This is all we need to construct $U$ with the desired properties; in fact $e_U(x,y)$ will be at least $0.5u(x,y)$, (and the factor $0.5$ is important for the proof).

Every machine $U$ has a ``dual'' representation: for every pair $(x,y)$ one may consider the subset $U_{x,y}$ of the Cantor space that consists of all $\pi$ such that $U(\pi,x)=y$ and $U(\pi,y)=x$. These sets are effectively open (i.e., are computably enumerable unions of intervals in the Cantor space) uniformly in $x,y$, are symmetric ($U_{x,y}=U_{y,x}$) and have the following property: for a fixed $x$, all sets $U_{x,y}$ for all $y$ (including $y=x$) are disjoint.

What is important to us is that this correspondence works in both directions. If we have some family $U_{x,y}$ of uniformly effectively open sets that is symmetric and has the disjointness property mentioned above, there exists a prefix-free machine $U$ that generates these sets as described above. This machine works as follows: given some $x$, it enumerates the intervals that form $U_{x,y}$ for all $y$ (it is possible since the sets $U_{x,y}$ are effectively open uniformly in $x,y$). One may assume without loss of generality that all the intervals in the enumeration are disjoint. Indeed, every effectively open set can be represented as a union of a computable sequence of disjoint intervals (to make intervals disjoint, we represent the set difference between the last interval and previously generated intervals as a finite union of intervals). Note also that for different values of~$y$ the sets $U_{x,y}$ are disjoint by the assumption. If the enumeration for $U_{x,y}$ contains the interval $[p]$ (the set of all extensions of some bit string~$p$), then we let $U(p,x)=y$ and $U(p,y)=x$ (we assume that the same enumeration is used for $U_{x,y}$ and $U_{y,x}$). Since all intervals are disjoint, the function $U(p,x)$ is prefix-free.

  Now it remains (and this is the main part of the proof) to construct the family $U_{x,y}$ with the required properties in such a way that the measure of $U_{x,y}$ is at least $0.5u(x,y)$. In our construction it will be \emph{exactly} $0.5u(x,y)$. For that we use the same idea as in theorem~\ref{th:bglvz} but in the continuous setting. Since $u(x,y)$ is lower semicomputable, we may consider the increasing sequence $u'(x,y)$ of approximations from below (that increase with time, though we do not explicitly mention time in the notation) that converge to $u(x,y)$. We assume that at each step one of the values $u'(x,y)$ increases by a dyadic rational number $r$. In response to that increase, we add to $U_{x,y}$ one or several intervals that have total measure $r/2$ and do not intersect $U_{x,z}$ and $U_{z,y}$ for any $z$. For that we consider the unions of all already chosen parts of $U_{x,z}$ and of all chosen parts of $U_{z,y}$. The measure of the first union is bounded by $0.5\sum_z u'(x,z)$ and the measure of the second union is bounded by $0.5\sum_z u'(z,y)$ where $u'$ is the lower bound for $u$ before the $r$-increase. Since the sums remain bounded by $1$ after the $r$-increase, we may select a subset of measure $r/2$ outside both unions. (We may even select a subset of measure $r$, but this will destroy the construction at the following steps, so we add only $r/2$ to $U_{x,y}$.)
\end{proof}

\begin{remark}
As for the other settings, we may consider two functions $U$ and $U'$ and the probability of the event 
  \[
e_{U,U'}(x,y) \;=\;\Pr_{\pi} [U(\pi,x)\mathop = y \text{ and } U'(\pi,y)\mathop = x]
\]
for those $U,U'$ that make this probability maximal. The equality of theorem~\ref{thm:apriori} remains valid for this version. Indeed, the easy part can be proven in the same way, and for the difficult direction we have proven a stronger statement with additional requirement $U=U'$. 
\end{remark}

One can also describe the function $\ee$ as a maximal function in some class, and we will explain that this provides a characterization of the maximum $\E$ of conditional complexities as an optimal function in some class. 

\begin{proposition}\label{prop:qpd}
Consider the class of symmetric lower semicomputable functions $u(x,y)$ with string arguments and non-negative real values such that $\sum_y u(x,y)\le 1$ for all $x$. This class has a maximal function that coincides with $\min(\mm(x\cnd y), \mm(y\cnd x))$ up to an $O(1)$ factor.
\end{proposition}

\begin{proof}
We have already seen that this minimum has the required properties; if some other function $u(x,y)$ in this class is given, we compare it with conditional semimeasures $\mm(x\cnd y)$ and $\mm(y\cnd x)$ and conclude that $u$ does not exceed both of them.
\end{proof}

In logarithmic scale this statement can be reformulated as follows: \emph{the class of upper semicomputable symmetric functions $D(x,y)$ with string arguments and real values such that $\sum_y 2^{-D(x,y)}\le 1$ for each $x$, has a minimal element that coincides with $\max(\KP(x\cnd y),\KP(y\cnd x))$ up to an $O(1)$ additive term}. Theorem 4.2 in~\cite{bglvz} says the same with the additional condition for~$D$: it should satisfy the triangle inequality. This restriction makes the class smaller and could increase the minimal element in the class, but this does not happen since the function
\[
\max (\KP(x\cnd y),\KP(y\cnd x))+c
\]
satisfies the triangle inequality for large enough $c$. This follows from the inequality $\KP(x\cnd z)\le \KP(x\cnd y)+\KP(y\cnd z)+O(1)$ since the left hand size increases by $c$ and the right hand size increases by $2c$ when $\KP$ is increased by $c$.

\begin{remark}
To be pedantic, we have to note that in~\cite{bglvz} an additional condition $D(x,x)=0$ is required for the functions in the class; to make this possible, one has to exclude the term $2^{-D(x,x)}$ in the sum (now this term equals $1$) and require that $\sum_{y\ne x} 2^{-D(x,y)}\le 1$ (p.~1414, the last inequality). Note that the triangle inequality remains valid if we change $D$ and let $D(x,x)=0$ for all $x$.
\end{remark}

\section{Proof of theorem~\ref{th:gap}}\label{sec:game}

For notational convenience, we first prove the following qualitative version of theorem~\ref{th:gap} for the non-bipartite distances. 

\begin{proposition}\label{prop:gap_qualitative}
  The difference between the non-bipartite prefix-stable distance and 
  \\ $\max(\KP(x \cnd y), \KP(y \cnd x))$ is unbounded.
\end{proposition}

\noindent
The quantitative statement can be easily obtained from the qualitative proof 
using a small calculation, given at the end of section~\ref{subsec:game-strategy}.
The modifications for the bipartite distances are also easy, and are explained in section~\ref{subsec:game-bipartite}. 
Together, this implies theorem~\ref{th:gap}.

\subsection{It is enough to win a game}\label{subsec:game-enough}

Consider the following two-player full information game. Fix some parameter $c$, a positive rational number. The game field is the complete graph on a countable set (no loops); we use binary strings as graph vertices.  Alice and Bob take turns. 

Alice increases \emph{weights} of the graph edges. We denote the weight of the edge connecting vertices $u$ and $v$ by $m_{u,v}$ (here $u\ne v$). Initially all $m_{u,v}$ are zeros. At her move, Alice may increase weights of finitely many edges using rational numbers as new weights. The weights should satisfy the inequality $\sum_{v\ne u} m_{u,v}\le 1$ for every $u$, i.e.,  the total weight of the edges adjacent to some vertex should not exceed~$1$.

Bob assigns some subsets of the Cantor space to edges. For all $u$ and $v \ne u$, the set $M_{u,v}$ assigned to the edge $u$--$v$ is a clopen subset of the Cantor space (clopen subsets are subsets that are closed and open at the same time, i.e., finite unions of intervals in the Cantor space). Initially all $M_{u,v}$ are empty. At each move Bob may increase sets assigned to finitely many edges (using arbitrary clopen sets that contain the previous ones). For every $u$, the sets $M_{u,v}$ (for all $v\ne u$) should be disjoint.

The game is infinite, and the winner is determined in the limit, assuming that both Alice and Bob follow the rules. Namely, Bob wins if for every $u$ and $v \ne u$, the limit value $\lim M_{u,v}$ (the union of the increasing sequence of Bob's labels for edge $u$--$v$) contains an interval in the Cantor space whose size is at least $c\cdot \lim m_{u,v}$ (the limit value of Alice's labels for $u$--$v$, multiplied by~$c$). Recall that the interval $[z]$ in the 
Cantor space is the set of all extensions of some string~$z$, and its size is $2^{-|z|}$. In the sequel, the size of the maximal interval contained in $X$ is denoted by~$\nu(X)$.

We claim that the existence of a winning strategy for Alice that is computable uniformly in~$c$, is enough to prove proposition~\ref{prop:gap_qualitative}. But first let us make some remarks on the game rules.

\begin{remark}
Increasing the constant $c$, we make Bob's task more difficult, and Alice's task easier. So our claim says that Alice can win the game even for arbitrarily small positive values of~$c$.
\end{remark}

\begin{remark}
In our definition the result of the game is determined by the limit values of $m_{u,v}$ and $M_{u,v}$, so both players may postpone their moves. Two consequences of this observation will be used. First, we may assume that Bob always has empty $M_{u,v}$ when~$m_{u,v}=0$. 
Second, we may assume that Bob has to satisfy the requirement $\nu(M_{u,v})\ge c \cdot m_{u,v}$ after each of his moves. Indeed, Alice may wait until this requirement is satisfied by Bob: if this never happens, Alice wins the game in the limit (due to compactness: if an infinite family of intervals covers some large interval in the Cantor space, a finite subfamily exists that covers it, too).
\end{remark}

\begin{lemma}\label{lem:gapGameImpliesTheorem}
  If Alice has a winning strategy in all games with~$c>0$ and if these strategies are uniformly computable in~$c$,
  then proposition~\ref{prop:gap_qualitative} is true.
\end{lemma}


\begin{proof}
Since the factor $c$ is arbitrary, we may strengthen the requirement for Alice and require $\sum_{v\ne u} m_{u,v}\le d$ for some $d>0$. This corresponds to the factor $cd$ in the original game. Given some integer $k>0$, consider Alice's winning strategy for $c=2^{-k}$ and $d=2^{-k}$. We  play all these strategies simultaneously 
against a ``blind'' strategy for Bob that ignores Alice's moves and just follows the optimal machine $U$ used in the definition of information distance. Here are the details.

Consider a prefix-stable computable partial function $U$ that makes the function 
$$
\D_U(u,v) \;=\; \min\{ |p|\colon U(p,u) \mathop{=} v \text{ and } U(p,v)\mathop{=}u\}
$$
minimal. For each edge $u$--$v$ consider the union of the sets $[p]$ for all $p$ such that $U(p,u)=v$ and $U(p,v)=u$ at the same time. This union is an effectively open set, and Bob enumerates the corresponding intervals and adds 
them to the label for the edge $u$--$v$ when they appear in the enumeration. Note that this set is the same for $(u,v)$ and $(v,u)$ by definition. For the limit set $M_{u,v}$ we then have $\nu(M_{u,v})\ge 2^{-\D_U(u,v)}$ by construction (consider the interval that corresponds to the shortest $p$ in the definition of $\D_U(u,v)$).

Let Alice use her winning strategy for $c=2^{-k}$ and $d=2^{-k}$ against Bob. Since Bob's actions and Alice's strategy are computable, the limit values of Alice's weights are lower semicomputable uniformly in $k$. Let us denote these limit values by $m^k_{u,v}$. We know that for every $u$ and $k$ the sum $\sum_{v\ne u} m^k_{u,v}$ does not exceed $2^{-k}$. Therefore the sum 
$$
 m_{u,v} = \sum_k m^k_{u,v}
$$
satisfies the requirement
$$
\sum_{v\ne u} m_{u,v} \le 1
$$
and we can apply proposition~\ref{prop:qpd}, where we let $m_{u,u}=0$. Recall that
$
\E(u,v) = \max\left(\KP(u \cnd v), \KP(v \cnd u)\right).
$ 
This proposition guarantees that
$$
m_{u,v} \le O(\min(\mm(u\cnd v),\mm(v\cnd u))=2^{-\E(u,v)+O(1)}.
$$
If, contrary to the statement of proposition~\ref{prop:gap_qualitative}, 
we have $\D_U(u,v) \le \E(u,v) + O(1)$, 
  then $\E(u,v)$ in the right hand side of the last inequality can be replaced by $\D_U(u,v)$. But this means, by our construction, that Bob wins the $k$th game for large enough~$k$, since the maximal intervals in $M_{u,v}$ are large enough to match $m_{u,v}$ (and therefore $m^k_{u,v}$) for large enough $k$, according to this inequality. 
We get a contradiction that finishes the proof of proposition~\ref{prop:gap_qualitative} for the non-bipartite case, assuming the existence of a uniformly computable winning strategy for Alice. 
\end{proof}

\begin{remark}\label{rem:gap_quantitative}
  The quantitative variant of the above lemma is as follows.
  If the winning strategy of the $k$-th game uses at most $N_k$ strings, 
  then difference between $\D_U - \E$ is at least~$k - O(1)$ on pairs of strings of length~$\lceil \log N_k\rceil$. 
\end{remark}

\subsection{How to win the game}\label{subsec:game-strategy}

Now we present a winning strategy for Alice. It is more convenient to consider an equivalent version of the game where Alice should satisfy the requirement $\sum m_{u,v}\le d$  and Bob should match Alice's weights without any factor, i.e., satisfy the requirement $\nu(M_{u,v})\ge m_{u,v}$. Assume $d$ is a negative power of $2$. 

The idea of the strategy is that Alice maintains a finite set of ``currently active'' vertices, initially very large and then decreasing. The game is split into $N$ stages where $N=1/d$. 
After each stage the set of active vertices and the edge labels satisfy the following conditions.
\begin{itemize}[leftmargin=*]
\item Alice has zero weights on edges that connect active vertices (as we have said, we may assume without loss of generality that Bob has empty labels on these edges, too).
\item For each active vertex, only a small weight is used by Alice on edges that connect it to other vertices (inactive ones; edges to active ones are covered by the previous condition and do not carry any weight); this weight will never exceed $d/2$.
\item More and more space is ``unavailable'' to Bob on each active vertex, 
  since it is already used on edges connecting 
    to inactive vertices. 
\end{itemize}
The amount of ``unavailable space'' for Bob grows from stage to stage until no more space is available and Alice wins. In fact, at each stage the amount of ``unavailable space'' grows by $d$, so Alice needs $N=1/d$ stages to make all space unavailable for Bob; then she makes one more request, i.e., increases a weight between an active vertex and a fresh one, and she wins, since Bob has no ``available space'' to fulfill this request.

In the previous paragraph we used the words ``unavailable space'' informally. What do we mean by unavailable space? Consider some active vertex $x$ and edges that connect it to inactive ones. These edges have some of Bob's labels, which are subsets of the Cantor space. The part of the Cantor space occupied by these labels is not available to Bob for edges between $x$ and other active vertices. Moreover, if Alice requests an interval of size $\eps$, i.e. increases some weight from $0$ to $\eps$, and some part (even a small one) of an interval of this size is occupied, then this interval cannot be used by Bob and is ``unavailable''. In this way the unavailable space can be much bigger than the occupied space, and this difference is the main tool in our argument.\footnote{This type of accounting goes back to G\'acs' paper~\cite{gacs1983} where he proved that monotone complexity and continuous a priori complexity differ more than by a constant, see also~\cite{usv} for the detailed exposition of his argument.} 

Let us explain this technique. First, let us agree that Alice increases only zero weights, and the new non-zero value of the weight
depends on the stage only. At the first stage she uses some very small $\eps_0$, at the second stage she uses some bigger $\eps_1$, etc. (so at the $i$th stage weights $\eps_{i-1}$ are used). We will use values of $\eps_i$ that are powers of~$2$ (since interval sizes in the Cantor space are powers of~$2$ anyway), and assume that $\eps_0\ll \eps_1\ll \eps_2\ldots$. More precisely, we let $\eps_N=d/2$ and assume that $\eps_{i-1}/\eps_i = d/2$. 
\begin{center}
\includegraphics[scale=1]{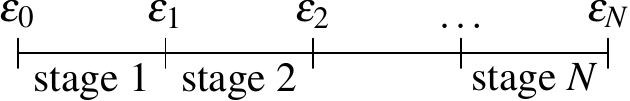}
\end{center}
This commitment about the weights implies that, starting from the $(i+1)$th stage, only the $\eps_i$-neighborhood of the space used by Bob matters. Here by $\eps$-neighborhood (where $\eps$ is a negative power of $2$) of a subset $X$ of the Cantor space we mean the union of all intervals of size $\eps$ that have nonempty intersection with $X$; note that the $\eps$-neighborhood of $X$ increases when $\eps$ increases (or $X$ increases).  

More precisely, let us call an interval \emph{dirty for 
vertex~$x$} (at some moment) if some part of this interval already appears in Bob's labels for edges that connect $x$ to other vertices. 
This interval cannot be used later by Alice. After stage $i$, we consider all the intervals of size $\eps_i$ that are ``everywhere dirty'', i.e., dirty for all active vertices (those that are dirty for some active vertices but not for the others, do not count). The everywhere dirty intervals form the \emph{unavailable space after stage $i$}, and the total measure of this space increases at least by~$d$ at each stage. In other terms, after stage~$i$ we consider for every active vertex $x$ the space allocated by Bob to all edges connecting $x$ with (currently) inactive vertices, and the $\eps_i$-neighborhood of this space. The intersection of these neighborhoods for all active vertices $x$ is the unavailable space after stage~$i$. 

After stage $i$ the total size of unavailable space will be at least $i/N$; recall that $N=1/d$. At the end, after the $N$th stage, we have $\eps_N=d/2$, so the total size of everywhere dirty intervals of size $d/2$ is $N/N=1$. Our strategy also implies that the total weight used by Alice at any vertex is~$d/2$. Finally, Alice makes one more request with weight $d/2$ and wins. Of course, we need that at least 1 vertex remains active after stage $N$, and this will be guaranteed if the initial number of active vertices is large enough.

The picture above places $\eps_i$ between stages since $\eps_i$ is used for accounting after stage $i$ and before stage $i+1$.

\bigskip
\noindent
It remains to explain how Alice plays at stage $i$ using requests of size $\eps_{i-1}$ and creating (new) everywhere dirty intervals of size $\eps_i$ with total size (=the size of their union) at least~$d$. This happens in several substages; each substage decreases the set of active vertices and increases the set of everywhere dirty intervals of size $\eps_i$ (for the remaining active vertices).

Before starting each substage, we look at two subsets of the Cantor space:
\begin{itemize}
\item[(a)] the set of intervals of size $\eps_{i-1}$ that were everywhere dirty after the previous stage;
\item[(b)] the set of intervals of size $\eps_i$ that are everywhere dirty now (after the substages that are already performed).
\end{itemize}
The second set is bigger for two reasons. First, we changed the granularity (recall the $\eps_i$-neighborhood of some set can be bigger than $\eps_{i-1}$-neighborhood). Second, the previous substages create new everywhere dirty intervals of size $\eps_i$. Our goal is to make the second set larger than the first one; the required difference in size is~$d$. If this goal is already achieved, we finish the stage (no more substage are necessary). If not, we initiate a new substage that creates a new everywhere dirty $\eps_i$-interval.

\begin{center}
\vbox{%
\begin{center}%
\includegraphics[scale=1]{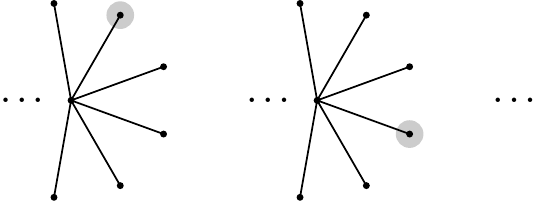}\\[1ex]
Alice's strategy for a substage
\end{center}%
}%
\end{center}

The key idea is that Alice makes requests for all edges of a large star. The center is a fresh vertex (all weights are zero), and the other vertices are active vertices. She may use a lot of weight for the central vertex, because the sum of the weights could be up to~$d$. 
Still for all other vertices of the star only one new edge of non-zero weight $\eps_{i-1}$ is added. Bob has to allocate some intervals of size at least $\eps_{i-1}$ for every edge in the star, and these intervals should be disjoint (due to the restrictions for the center of the star). 
The total measure of these intervals is $d$, and all of them are outside the zone (a). Therefore, since the goal is not yet achieved, one of these new intervals used by Bob is also outside the zone (b). Corresponding neighbors of the central vertex are indicated with a gray circle in the above picture.

Alice does the same for many stars (assuming that there are enough active vertices) and gets many new $\eps_i$-intervals outside the (b)-zone (at least one per star). Some of them have to coincide: if we started with many stars, we may select many new active vertices that have the same new $\eps_i$-dirty interval. Making all other vertices inactive, we get a smaller (but still large if we started with a large set of active vertices) set of active vertices and a new everywhere dirty $\eps_i$-interval. The goal of a substage is achieved. We look again at the set of everywhere dirty $\eps_i$-intervals (with the new intervals added) to decide whether the difference between (b) and (a) is at least $d$, or whether a new substage is needed. The maximal number of substages needed to finish the stage is $d/\eps_i$, since each substage creates a new $\eps_i$-interval.

The same procedure is repeated for all $N$ stages. We need to check that Alice uses at most $d/2$ weight connecting some active vertex to all inactive vertices. For that, we look at the ``amplification factor''. In the construction Alice uses a single weight $\eps_{i-1}$ (for every active vertex) to get a new dirty interval of size $\eps_i$, therefore the amplification factor is $\eps_i/\eps_{i-1}=2/d$. Since the total size of dirty intervals is at most $1$, the total weight used by Alice (for each active vertex) never exceeds $d/2$, as required.

It remains to explain why Alice can choose enough active vertices in the beginning, so she will never run out of them in the construction and at least 1 vertex remains active at the end (so the last request of size $d/2$ wins the game). Indeed, the backwards induction shows that for each substage of each stage there is some finite number of active vertices that is sufficient for Alice to follow her plan till the end. If we want to upper bound the length on the strings where a given value of the difference~$\D_U-\E$ is achieved, we need to compute this number explicitly. 
But the qualitative statement of  proposition~\ref{prop:gap_qualitative}, (the unbounded difference) is already proven. 

\bigskip
\noindent
We now perform this computation. 
As explained in remark~\ref{rem:gap_quantitative}, we need to compute the number of different strings for which a strategy with 
$d = 2^{-2k}$ makes (recall that in lemma~\ref{lem:gapGameImpliesTheorem} we used strategies with $c = d = 2^{-k}$, but above we assumed $c = 1$, corresponding to $d = 2^{-2k}$). 
The logarithm of this number provides us with the length for which the gap is at least $k-O(1)$. (With a more careful analysis we could obtain $2k - O(\log k)$, but this does not matter for the statement of  theorem~\ref{th:gap}.) 

Recall that $N = 1/d$. We have $1/\eps_0 \le 2^{O(N^2)}$ by the choices $\eps_N = d/2$ and $\eps_{i-1}/\eps_i = d/2$. 
There are $1/\eps_0$ intervals of size $\eps_0$ and each star contains $2N$ active strings, thus the fraction of active vertices that remain after the first substage is $\eps_0 /(2N) \le O((\eps_0)^2)$.
The number of substages is $d/\eps_0 \le 1/\eps_0$, hence the total fraction of active vertices that are lost during the first stage is at most ${\eps_0}^{O(1/\eps_0)}$. The fraction of vertices lost in future stages increases double exponentially, and hence the same expression determines the total number of active vertices that we need to start with (in order to end with at least 1 active vertex). These vertices can be associated to strings of length 
\[
  O(1/\eps_0) \cdot \log (1/\eps_0) \;\le\; 2^{O(N^2)} \;\le\; 2^{O(2^{4k})}.
\]
Taking twice the logarithm, we conclude that the difference on $n$-bit strings is at least $\tfrac 1 4 \log \log n - O(1)$. 
We have proven theorem~\ref{th:gap} for the prefix-stable non-bipartite case. 
The prefix-free case is a corollary (the distance becomes bigger), but for the bipartite case we need to adapt the argument, and this is done in the next section.

\begin{remark}\label{rem:lowerbound}
  With a more careful argument, the difference in theorem~\ref{th:gap} can be shown to be at least~$\log \log n - O(\log \log \log n)$.
\end{remark}

\subsection{Modifications for the bipartite case}\label{subsec:game-bipartite}

In the bipartite case the game should be changed. Namely, we have a complete bipartite graph where left and right parts contain all strings. Alice increases weights on edges; for each vertex (left or right) the sum of the weights for all adjacent edges should not exceed some $d$ (the parameter of the game). As before,  Alice increases weights $m_{x,y}$, and at each moment these weights are symmetric, i.e., $m_{x,y} = m_{y,x}$. In our strategy, this requirement will not matter, because we will only increase weights for pairs $(x,y)$ in a product set $X \times Y$ with disjoint sets $X$ and $Y$. Thus, we drop the requirement of symmetry and require instead that
$$
\forall x\, \left(\sum_y m_{x,y}\le d\right) \quad \text{and} \quad
\forall y\, \left(\sum_x m_{x,y}\le d\right).
$$

Bob replies by assigning increasing sets $M_{x,y}$ to edges such that $\nu(M_{x,y})\ge m_{x,y}$. For each $x$ the sets $M_{x,y}$ (with different $y$) should be disjoint; the same should be true for sets $M_{x,y}$ for fixed $y$ and different $x$. The sets $M_{x,y}$ and $M_{y,x}$ can be different, (but this does not matter, since our strategy uses pairs in a product $X \times Y$ of disjoint sets).

Again, to prove that the bipartite prefix-free information distance exceeds $\E(x,y)$ 
by a constant, we show that for every $d$ Alice has a computable (uniformly in $d$) winning strategy in this game. Then we consider games with total weight $2^{-k}$ and condition $\nu(M_{x,y})\ge 2^{-k}m_{x,y}$. We let Alice play her winning strategy against the ``blind'' strategy for Bob that (for the edge $x$--$y$) enumerates all intervals $[p]$ such that $U(p,x)=y$ and $U'(p,y)=x$ at the same time.

\begin{figure}
\begin{center}
\includegraphics[scale=0.8]{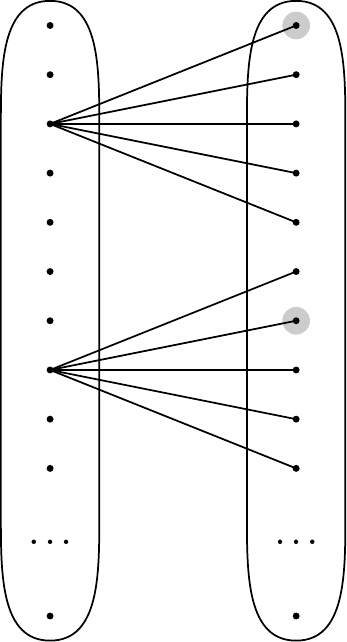}
\end{center}
  \caption{Stars in the bipartite strategy} \label{fig:bipartite}
\end{figure}

The winning strategy for Alice works in almost the same way.  Alice keeps the list of active vertices 
on the right and the centers of the star are chosen on the left, see figure~\ref{fig:bipartite}.  
As before, she uses fresh strings for these centers.
Thus only weights on edges $(x,y) \in X \times Y$ are increased, where the set of centers $X$ is disjoint from the set $Y$ of active vertices.

In each center of a star the sum of Alice's weights is $d$, and this implies the left condition on~$m_{x,y}$ above. 
In each right node, the sum of all weights is $d/2$ by the same density argument as before, and this implies the right condition.
After the last stage, there is an active vertex for which all intervals of size $\eps_N = d/2$ are dirty. 
Alice wins by making a final request.
The quantitative analysis does not change, and theorem~\ref{th:gap} for the bipartite distances is proven.

%

\section{Proof of theorem~\ref{th:exact} }\label{sec:exact}

Recall that $ \E(x,y) = \max \left( \KP(x \cnd y), \KP(y \cnd x)\right)$.  We restate the theorem.

\begin{theorem*}
  \theoremExact
\end{theorem*}

\subsection{It is enough to win a game}\label{ss:gameExact}

We only need to prove the theorem for the largest of the 4 distances, 
which is the non-bipartite prefix-free information distance. 
We first present the game that corresponds to this distance.
It is similar as before. 
The main differences are that it is played on strings of a fixed length and Bob's requirement 
must only hold for small weights.
This time, we need a winning strategy for Bob.

\medskip
\noindent
The game has 3 parameters: a positive integer $n$, a real number $d>0$ and a real number $\delta > 0$.
It is played on a graph with $2^n$ vertices that are labelled by $n$-bit strings.
Initially these weights are all~$0$. 

At her turn, Alice makes  {\em requests} of the form~$(\{u,v\},\eps)$ with $\eps \le \delta$, 
(in the language of the previous game, this means that she increases weight $m_{u,v}$ by~$\eps$). 
For each vertex $u$, the sum of request sizes of requests of the form $(\{u, \cdot\},\cdot)$ 
should be at most $d$, (this is the same as the requirement~$\sum_v m_{u,v} \le d$ in the previous game). 

As before, for each edge~$(u,v)$ Bob maintains a subset $M_{u,v}$ of the Cantor space. 
Initially, these sets are empty.
Bob enumerates basic intervals $[q]$ in $M_{u,v}$. 
He may only enumerate an interval in $M_{u,v}$ that
does not intersect the current set~$M_{u,v}$ and
does not intersect $M_{u,v'}$ for each $v' \not= v$.


In order to win, for each request $(\{u,v\},\eps)$ Bob should enumerate an interval
interval~$[q]$ in~$M_{u,v}$ of size at least~$\eps$. 
If he did not already do this before the request, he should do it immediately after the request, and if he fails to do so, he looses.

\begin{remark}
  In the next lemma, it is shown that the theorem follows from a winning strategy for Bob for some $d>0$. 
  Thus a constant factor in $d$ is not important for us, (it is absorbed in the additive $O(1)$ term of the equality in the theorem).
  One might be interested to compare Bob's requirement to the one in the previous section, and we discuss 3 differences. 

  Firstly, we compare the interval size of a reply not to the total weight, 
  but only to the size of the last increase. 
  But this does not matter, because we can postpone Alice's moves so that for a fixed edge, the request sizes increase geometrically. 
  After this, the requirements differ by a constant factor which is equivalent to a constant rescaling of~$d$.

  The second difference, is that previously we considered winning in the limit, 
  but now we require that Bob should give the required reply immediately.
  As said before, this does not matter because
  Alice may postpone her moves as long as she is in a winning position.

  The last difference is that previously, we considered the maximal size $\nu(M_{u,v})$ of an interval that is a subset of~$M_{u,v}$, 
  and this means that a largest interval could be gradually increased.
  Now we consider the maximal size of an interval $[q]$ in a single reply. 
  (Thus, if a reply is needed, all previously allocated intervals in $M_{u,v}$ are irrelevant.)
  This is because we prove an upper bound for the prefix-free distance, which is the largest distance, 
  while previously, we considered the prefix-stable distance.
  We do not know whether these requirements are equivalent when the restriction $\epsilon \le \delta$ is removed. 
  If they are equivalent, then the prefix-stable and the prefix-free distances are always the same, and this is an open question. 
  (With requests of size at most~$\delta \le 1.01 \log n$, the theorem implies that these games are indeed equivalent up to constant rescaling of~$d$, 
  but we do not know how to prove this more directly.)
\end{remark}


\begin{lemma}\label{lem:strategyImpliesTheoremExact}
  If there exist $d>0$ and a function $\delta$ such that $\delta(n) \ge \Omega(n^{-1.01})$ for all $n$ and such that Bob has a winning strategy in the above game,
  then theorem~\ref{th:exact} is true.
\end{lemma}

\begin{remark}\label{rem:makeGameFinite}
  The game tree is infinite because request sizes $\eps$ can be real numbers. 
  However, we explain that the assumption of the lemma does not change if we use a finite variant of the game.
  In other words, the existence of a winning strategy does not change up to a constant rescaling of~$d$.
  
  Recall from the previous remark, that we may assume that Alice's request sizes increase geometrically.
  For example, we may assume she uses only negative powers of~2. Now consider the game in which Alice's requests have size at least~$2^{-n-2}$.
  Imagine that Bob connects all pairs of strings with intervals of size~$2^{-n-2}$.
  He can do this using one half of the Cantor space. 
  On the other half, he plays a scaled version of a winning strategy for the restricted game (with request sizes larger than~$2^{-n-2}$).
  Thus after a decrease of $d$ by a factor 2, he also wins in the unrestricted game. 

  Hence, the game is finite, and a winning strategy can be computed given the value of the parameters, for example, by exhaustive search.
  Note that we may assume that Alice uses at most $m \le n + 2$ different request sizes.
\end{remark}

\begin{proof}[Proof of lemma~\ref{lem:strategyImpliesTheoremExact}.] 
  We may assume $\delta(n) = d\cdot n^{-1.01}$, because if $\delta(n)$ is smaller (by a constant factor), 
  we decrease~$d$.
  We also use a family of winning strategies that can be computed uniformly in~$n$, 
  which exists by the remark above. 

  We construct a prefix-free machine $V$ such that for all $n$ and all pairs $(x,y)$ of different $n$-bit strings,
  there exists a program $q$ for which $V(q,x) = y$, $V(q,y) = x$ and $|q| \le \E(x,y) + \log \tfrac{1}{d}$, 
  provided $\E(x,y) \ge 1.01 \log n$.

  \medskip
  \noindent
  {\em Construction of $V$.}  On  input $(q,x)$, machine $V$ runs
  Bob's winning strategy for strings of length $n = |x|$ 
  against Alice's strategy in which she has limit weights $m_{u,v} = d2^{-\E(u,v)}$. 
  More precisely, she uses an (integer) approximation of~$\E$ from above, 
  and for each update of an approximated value of~$\E(u,v)$ to a value~$k$, 
  she generates a request of the form $(\{u,v\},d2^{-k})$, provided that $d  2^{-k} \le \delta(n)$.
  Bob replies by enumerating intervals $M_{u,v}$. If for some $v$, 
  the interval $[q]$ is enumerated in some set~$M_{x,v}$, 
  then $V(q,x)$ halts with output~$v$.

  \medskip
  \noindent
  Note that Alice's requests indeed satisfy the restriction $\sum_v m_{u,v} \le d$, since $\sum_v 2^{-\E(u,v)} \le 1$.
  For every pair $(x,y)$ with $\E(x,y) \ge 1.01\log n$, Alice makes requests of size at most $d 2^{-1.01\log n} = \delta(n)$.
  The winning condition implies that at some point an interval $[q]$ 
  with $|q| \le \log \tfrac{1}{d} + \E(x,y)$ is enumerated in $M_{x,y}$. 
  By construction, this implies $V(q,x) = y$ and $V(q,y) = x$. 
  Thus the information distance defined by $V$ is at most~$\E(x,y) + O(1)$.
 
  We show that $V$ is a prefix-free machine. 
  For a fixed~$y$ and different $u$, the sets $M_{u,y}$ are disjoint.
  Moreover, each time a new set $[q]$ is enumerated in $M_{u,y}$, 
  it does not overlap with intervals that were previously enumerated in $M_{u,y}$.
  (The last requirement is not needed for the prefix-stable distance.)
  Thus, for each $y$, the set of programs $p$ for which $V(p,y)$ halts, is prefix-free.
  The prefix-free non-bipartite information distance defined for machine $V$ satisfies the conditions of theorem~\ref{th:exact}, and 
  hence the same holds when defined with an optimal machine. 
\end{proof}

\subsection{Easy strategies for Bob}\label{ss:easyStrat}

Recall that Alice's winning strategy in the previous section was to force Bob's allocations to be spread,
i.e., somewhat smaller requests are uniformly distributed over the Cantor space, 
and no space is available for large requests. 
A winning strategy for Bob must be able to avoid this. 
In other words, Bob must ensure that requests of the same size are allocated in a few contiguous intervals of the Cantor space.
As we know from Alice's winning strategy above, Bob can only do this for requests that are sufficiently small. 

We first consider variants of the game where Alice satisfy strong restrictions and where Bob can easily localize requests of the same size.
Afterwards, we weaken these restrictions gradually.
We always assume that requests sizes $\eps$ are are negative powers of~$2$.
The first variant is the most basic one, and the strategy is obtained from the proof of the characterization of the plain distance.


\medskip
\noindent
\textbf{Variant A.} 
Alice can only make requests of a fixed size~$\eps$.

\medskip
\noindent
 {\em Winning strategy for $d=\tfrac 1 2$.} 
Bob can win using a greedy strategy. 
Indeed, the game is equivalent to the edge-coloring problem in the proof 
of  theorem~\ref{th:bglvz}, because increasing a weight from $0$ to $\eps$ 
corresponds to enumerating an edge in $R_{\log (1/\eps)}$, and 
colors can be associated to intervals of size~$\eps$.

\bigskip
\noindent
In the following example we show that if for each string, the distribution of request sizes is the 
same, then Bob can win by playing several copies of this greedy strategy in parallel.

\medskip
\noindent
\textbf{Variant B.} 
Both players are given sizes $\eps_1, \ldots, \eps_m$ and positive integers $N_1, \ldots, N_m$ such that $\eps_1 N_1 + \ldots + \eps_m N_m \le d$. 
For each string $u$ and each size $\eps_i$, Alice can make at most $N_i$ requests of the form~$(\{u,\cdot\}, \eps_i)$.

\medskip
\noindent
{\em Winning strategy for $d = \tfrac 1 4$.} 
We provide a strategy for $d = \tfrac 1 2$ assuming that all $N_i$ and hence $\eps_iN_i$ are powers of~2. 
This is enough, because rounding up may at most double the sum
$\eps_1 N_1 + \ldots + \eps_m N_m$.
To each string $u$, we associate a copy $\Omega_u$ of the Cantor space. 
If Bob replies by enumerating an interval $[q]$ in $M_{u,v}$, we say that he {\em allocates} 
$[q]$ both in $\Omega_u$ and $\Omega_v$. 
In the beginning of the strategy, Bob divides each Cantor space $\Omega_u$ in regions, 
and associates each size $\eps_i$ to a region of size~$2\eps_iN_i$.
All Cantor spaces $\Omega_u$ are partitioned in the same way.
When given a request of size $\eps_i$, Bob uses the above greedy strategy inside the corresponding region.
This strategy works inside each region, and hence Bob wins.

\bigskip
\noindent
We now consider the same variant as above, but Alice can choose the values $N_1, \ldots, N_m$ during the game.

\medskip
\noindent
\textbf{Variant C.} 
The players are given $\eps_1, \ldots, \eps_m$.
Let $N_{i,u}$ be the number of requests of the form $(\{u,\cdot\}, \eps_i)$ during the whole game. 
Alice's requests should satisfy: 
\[
  \sum_{i\le m} \eps_i \cdot \max_{u}   N_{i,u}   \;\le\;   d\,.
  \]

\medskip
\noindent
 {\em Winning strategy for $d = \tfrac 1 4$ if $m$ is a power of $2$.} 
Bob creates $2m$ regions of equal size in the Cantor space. 
The idea is to play a similar strategy as above and dynamically associate request sizes to regions when needed.
Initially, all regions are unassigned.
Given a request $(\{u,v\},\eps_i)$ Bob searches for an interval inside the regions associated to~$\eps_i$ that is free both for $u$ and~$v$.
If such an interval exists, it is allocated and we are finished.
Otherwise, he associates a new region to the request size~$\eps_i$, and allocates some interval in it.

What is the maximal number~$r$ of assigned regions that can appear?
Let $r_i$ be the number of regions associated to size~$\eps_i$. Thus $\sum_i r_i = r$.
Assume a request $(\{u,v\},\eps_i)$ can not be allocated. Then all the measure of regions associated to $\eps_i$ is either 
allocated in $u$ or in $v$. Thus $\max(N_{i,u}, N_{i,v}) \ge \tfrac 1 2 \cdot \tfrac{r_i}{2m}$. 
When assigning a new region, the value of $r_i$ increases by 1. Thus at any moment we have
\[
\eps_i  \max_u N_{i,u} \;\ge\; (r_i - 1) \cdot \tfrac{1}{2} \cdot \tfrac{1}{2m}.
\]
Summing over $i$, we obtain $d \ge (r-m)/(4m)$.
Thus for $d=\tfrac 1 4$, at most~$r = 2m$ regions can be assigned, and the strategy can always assign a new region when needed.

\subsection{Blaming strings for failed requests}\label{ss:givingBlame}

In the previous subsection, the distribution of request sizes involving a fixed string, is the same for each string. 
Now we allow this distribution to vary. For example, one string may only receive requests of size $\eps_1$, while another 
string only receives requests of size~$\eps_2$. 
In this case, we can no longer associate request sizes to regions in the same way for all strings.

\medskip
\noindent
\textbf{Variant D.}
The players are given $\eps_1, \ldots, \eps_m$
and for each string $u$, a list $N_{u,1}, \ldots, N_{u,m}$ of non-negative integers with $\sum_i \eps_i N_{u,i} \le d$.
For each $u$ and size $\eps_i$, Alice may make at most $N_{u,i}$ requests of the form $(\{u,\cdot\}, \eps_i)$.

\medskip
\noindent
For some choices of $\eps_i$ and $N_{u,i}$, Alice has a winning strategy.
Indeed, in Alice's winning strategy in the proof of theorem~\ref{th:gap}, requests sizes can be fixed in advance.
Therefore, the existence of a winning strategy for Bob requires that all request sizes $\eps_i$ are sufficiently small.

All strategies that we consider below involve regions and we can always use the technique from variant C to associate 
sizes to regions dynamically for each individual string (more details below). 
Therefore, knowing all $N_{u,i}$ in advance does not help Bob.
After dropping the requirement and using request sizes $2^{-n-2}, 2^{-n-1}, \ldots, 2^{\lfloor \log \delta\rfloor}$, 
we obtain a variant to which lemma~\ref{lem:strategyImpliesTheoremExact} can be applied, see remark~\ref{rem:makeGameFinite}.

As an intermediate step, we now consider another way to simplify the task for Bob: we allow him to ignore a few requests.
If he does not make an allocation for a request, 
he must blame one of the strings of the request.

\medskip
\noindent
\textbf{Variant E.}
The players are given $\eps_1, \ldots, \eps_m$ and an integer~$T$.
Given a request $(\{u,v\}, \eps_i)$, Bob must either allocate an interval, 
{\em blame}~$u$, or {\em blame}~$v$. 
During the whole game, a given string may be blamed at most~$T$ times.

\begin{proposition}\label{prop:strat_blaming}
  Assume $d$ is small, $c$ is large and $\eps_i \le 1/(cm^3n)$ for all $i \le m$.
  Bob has a strategy in \textnormal{variant E} in which each string is blamed at most $O(m^2)$ times.
\end{proposition}


\subsection{Winning strategy for Variant E}

We first describe the strategy and explain informally why it works. 
Afterwards, we present formal definitions and a combinatorial lemma from which the winning condition follows. 

As in variant C, we use $r = 2m$ regions. Again, regions of a string are disjoint and used for a single request size. 
But now the regions are assigned differently for different strings. Moreover, each region of a string
intersects with each region of any other string.  
More precisely, we partition the Cantor space in a large (but polynomial in $r$) number of blocks of equal size. 
Blocks are assigned to one of $r$ regions randomly and independently for each string. 
(In fact, we will use a deterministic assignment that satisfies the conditions of lemma~\ref{lem:extractorlike}. 
However, its existence follows by showing that some random assignment satisfies the conditions with positive probability.)
Thus each region is equal to the union of blocks that are assigned to~it.

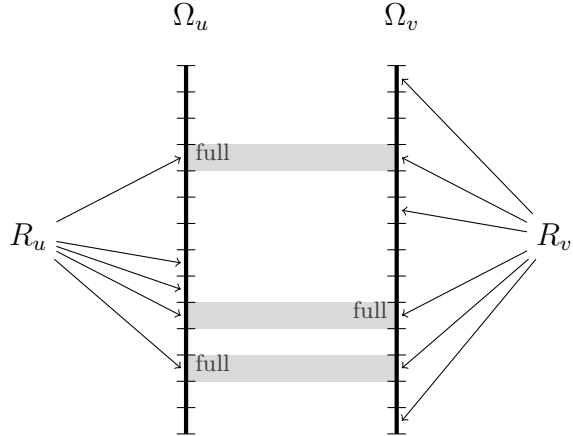
\begin{figure}[h]
  \centering
  \begin{tikzpicture}[scale=0.7]
    \foreach \sym/\sgn/\nameString/\symm in {A/-/u/u,B//v/v}
    {
      \node (R\symm) at (\sgn 5,0) {\large $R_{\symm}$};
      \foreach \x in {0,1,2,...,14}
	  {
	    \coordinate (\sym\x) at ($(\sgn 2,-3.75)+(0,\x*0.5cm)$) {};
	  }
      \node[above=1em] at (\sym14) {\large{$\Omega_\nameString$}\!\!};
      }

      \foreach \x in {2,4,5,6,10}
      {
	\draw[->] (Ru) -- ($(A\x)+(-1mm,2.5mm)$);
      }
      \foreach \x in {0,2,4,8,10,13}
      {
	\draw[->] (Rv) -- ($(B\x)+(1mm,2.5mm)$);
      }

      \foreach \leftBot/\rightUp in {A2/B3, A4/B5, A10/B11}
      \draw[draw=none,fill=gray!30] (\leftBot) rectangle (\rightUp);

      \foreach \pos/\direction in {A2/west, A10/west, B4/east}
      \node[yshift=2.5mm,anchor=\direction,darkgray] at (\pos) {\small{full}};

      \foreach \sym in {A,B}
      {
	\foreach \x in {0,1,...,14}  \draw ($(\sym\x)+(5pt,0pt)$) -- ($(\sym\x)-(5pt,0pt)$);
	\draw[ultra thick] (\sym0) -- ($(\sym14)$);
      }
  \end{tikzpicture}
  \caption{Illustration of Bob's strategy.
  $R_u$ is a region of~$u$ consisting of 5 blocks, 
  and $R_v$ is a region of~$v$ consisting of 6 blocks. 
  If no interval can be allocated, the strategy inspects the common blocks, 
  indicated in gray. 
  The top and bottom block are full for $u$ and the middle block is full for~$v$, 
  thus~$u$ is blamed.
  }\label{fig:strategyExact}
\end{figure}

The expected measure of a region is~$1/r$, (each block belongs to the region with the same probability).
The Chernoff bound implies that with high probability
each region has measure equal to~$1/r$ up to a factor $1 + \xi$ for some small~$\xi>0$.
Similarly, consider 2 regions $R_u$ and $R_v$ of different strings $u$ and $v$, see figure~\ref{fig:strategyExact}. 
With high probability, the total measure of all blocks belonging both to $R_u$ and $R_v$, equals~$1/r^2$ up to a factor $1+\xi$.
In particular, each such pair of regions has at least 1 block in common.
%

\medskip
\noindent
{\em Bob's Strategy to allocate an interval for the request $(\{u,v\},\eps)$.} 
Select a region $R_u$ of~$u$ to which the request size $\eps$ is associated
and for which the fraction of full blocks is less than~$1/4$. 
If no such region exists, select a fresh region of~$u$ and associate the request size~$\eps$. 
Similarly for the selection of a region~$R_v$ of~$v$. 

If there exists a free interval of size $\eps$ in the intersection of $R_u$ and $R_v$, 
then allocate an interval. 
Otherwise,  {\em blame} a string in $\{u,v\}$ for which at least half of the common blocks of~$R_u$ and~$R_v$ are full.
(Note that in this case, all common blocks are either full for $u$ or for~$v$, thus one of these strings is indeed blamed.)
{\em End of the strategy.}

\medskip
\noindent
See figure~\ref{fig:strategyExact} for an illustration. 
We first explain that for small $d$, at most $2m$ regions are assigned for each string. 
This is similar as for variant C above.
At most $m$ regions can be non full, i.e., have less than $1/4$-th fraction of full blocks, 
because a new region is only assigned when all other regions are full.
For small $d$, at most half of the regions can be full, since they have approximately equal size. 
Thus, at most $m+ (2m)/2 = 2m$ regions are assigned.

We qualitatively explain that with positive probability over the random assignment of blocks to regions, the following holds: 
for each strategy of Alice, each string is blamed only a few times.
Consider a fixed string $u$ and a region $R_u$ for requests of size~$\eps$. 
Let $B$ be the set of all regions $R_v$ that were used by the strategy when $R_u$ 
was selected and $u$ was blamed. 
We must explain why $B$ is small, and thus $u$ is blamed only a few times.
Note that at the end of the game, at most a fraction~$1/4$ of blocks in a region can be full, 
because the region can only be selected when this fraction is smaller.

For the sake of contradiction, let us assume that $B$ is (very) large.
Imagine we do the following experiment: select a random region $R_v$ in $B$, and 
randomly select a common block of~$R_u$ and~$R_v$. We determine the probability of the event 
``the block is full for $u$'' in 2 different ways.
On one hand, since $B$ is large, we expect that the obtained block is close to uniformly 
distributed in~$R_u$, because blocks are randomly assigned to regions, 
so 
the average distribution must be close to uniform, (random functions are good extractors).
Thus the probability that the block is full for~$u$ is at most~$1/4$ plus a small amount, 
since the fraction of full blocks of $R_u$ is at most~$1/4$.
On the other hand, recall that $u$ is blamed, thus for each $R_v \in B$, half of the common blocks are full for~$u$.
Thus, with probability at least~$1/2$ the resulting block must be full for~$u$. 
This is a contradiction, and our assumption must be false: $B$ can not be too large.

\bigskip
\noindent
We present the combinatorial requirements for the assignment that guarantees that the above set $B$ has size~$O(m)$, 
and hence each string is blamed at most $O(m^2)$ times, since a set $B$ exists for each region of~$u$ and there are $O(m)$~regions.

Recall that we partition the Cantor space in blocks of equal size. 
Let $\ell$ represent the number of such blocks. 
Also recall that regions are subsets of the Cantor space consisting of the union of all blocks with indices in some set $I \subseteq [\ell] = \{1,2,\ldots, \ell\}$. 
Such a set~$I$ is called an {\em index set}.
Let $\mcI$ represent a list of $r2^n$ subsets of $[\ell]$, where $r$ represents some upper bound on the number of regions that we use for each string. 

We will use a slightly different variant of the strategy, where regions are assigned in an online way when needed.
This extra feature is important in section~\ref{ss:blamingFriends}.\footnote{
  It also allows for a nicer combinatorial lemma. 
  But unfortunately, it requires some extra technical steps in the proof of proposition~\ref{prop:strat_blaming}.
  A proof with offline assignments can be found in earlier versions on ArXiv.
} 
For this, we use an upper bound $r$ that is larger than $2m$ but still satisfies $r \le O(m)$.
Each time Bob assigns a new region for some string~$u$, he selects an unused 
index set from the list~$\mcI$, and creates the region $R_u$ given by the union of the {\em unused} blocks in the index set, 
(here, {\em used} blocks are blocks that belong to previously assigned regions of $u$, and they must be excluded).
The following lemma provides a list~$\mcI$ of index sets for which Bob's strategy satisfies the conditions of proposition~\ref{prop:strat_blaming}.

\newcommand{\combLemmaOne}{
Let $\xi > 0$ be small and $e$ be large. For all $r, N$ and $\ell \ge er^3\log(N)$, there exists a list $\mcI$ of $N$ subsets of $[\ell]$ such that:
  \begin{itemize}[leftmargin=*]
    \item 
      Up to a factor $1+\xi$, each triple $(I,J,K)$ of different items in $\mcI$ satisfies:
      \[
       \# I = \ell/r, \qquad \# (I \cap J) = \ell/r^2 \qquad and\qquad \# (I \cap J \cap K) = \ell/r^3.
     \]

    \item 
      For each item $I$ of $\mcI$ and for each $I' \subseteq I$ of size $\tfrac{1}{4}\#I$, there exist at most $O(r)$ items $J$ in $\mcI$ for which
      \[
	\# \Big( I' \cap J \Big) \;\ge\; (\tfrac{1}{2}-3\xi) \;\# \Big( I \cap J \Big).
      \]
  \end{itemize}
  }
\begin{lemma}\label{lem:extractorlike}
  \combLemmaOne
\end{lemma}

\noindent
The proof of this lemma is given in section~\ref{ss:combinatorial}.

\begin{proof}[Proof of proposition~\ref{prop:strat_blaming}.]
  Let $r = (4/\xi) \cdot m$, where $\xi$ is the constant of  lemma~\ref{lem:extractorlike}. 
  We assume $\xi \le 1/8$.
  Let $N = r2^n$, and let $\ell$ be the smallest power of two that satisfies the bound of lemma~\ref{lem:extractorlike}.
  Let~$c$ in the assumption of proposition~\ref{prop:strat_blaming} be large enough such that $\eps_i \le 1/2\ell$ for all~$\eps_i$.
  We apply Bob's strategy as explained above using the list $\mcI$ from the lemma. 

  \medskip
  \noindent
  The first step is technical. We show that {\em for each string, Bob's strategy assigns at most~$\xi r$ regions}.
  Thus, in total we assign at most $\xi N \le N$ regions, and hence $\mcI$ contains enough index sets.

  Recall that a region is  {\em full} if at most a quarter of its blocks are full.
  The total measure of full regions must be smaller than~$d/(\tfrac 1 4 \cdot \tfrac 1 2) = 8d$ and at most~$m$ regions can be non-full.
  Hence, the total measure of assigned regions is at most $8d + \xi/4$. 
  For small~$d$, this is at most~$\xi/2$.
  We prove that each assigned region has measure at least~$1/(2r)$. 
  This implies our goal, because at most $\tfrac {\xi/2} {1/2r} = \xi r$ regions can be assigned.
  
  For the first assigned region, this follows from $\# I/\ell \ge (1-\xi)/r \ge 1/(2r)$ for each index set~$I$ in~$\mcI$.
  For the sake of induction, suppose that all previously assigned regions have measure at least~$1/(2r)$. 
  As we already explained, this implies that at most $\xi r$ regions have previously been assigned.
  We need to subtract the used blocks from a fresh index set. 
  By the first item of the lemma, each assigned region overlaps in a measure at most $(1+\xi)/r^2$. 
  Therefore, each region has measure at least
  \[
    \tfrac{1-\xi}{r} \;-\; (\xi r) \cdot \frac{1+\xi}{r^2}.
  \]
  This is at least $r/2$ for small~$\xi$, and in particular for~$\xi \le 1/8$. 
  The first step is finished.

  \medskip
  \noindent
  It remains to show that each string is blamed at most~$O(r^2)$ times. 
  We say that a region $R_u$ is blamed if in the strategy, $R_u$ is selected and string~$u$ is blamed.
  Since there are $O(r)$ regions, it suffices to show that region~$R_u$ is blamed at most $O(r)$ times.
  We prove this using the second item of the lemma.
  Let $I$ be the set in $\mcI$ that was used to assign region~$R_u$, 
  and let $I' \subseteq I$ be the set of full blocks of the region at the end of the game.
  If a region is selected, less than a quarter of its blocks are full, 
  thus $\#I' \le \tfrac{1}{4}\#I$.
  Each time the region is blamed, the region $R_v$ corresponds to an item $J$ in $\mcI$ that satisfies the inequality of the second item of the lemma.
  Here, the negative $3\xi$ term compensates for used blocks that were removed during the assignment.
  We explain that this term is enough. 

  Indeed, removing indices of a set $K$ from either $I$ or $J$, decreases the number of elements in $I \cap J$, 
  by at most $\# (I \cap J \cap K) \le (1+\xi)\ell/r^3$.
  The number of such $K$ is at most $2 \cdot (\xi r)$, by the first step of the proof. 
  Thus we remove at most $2(1+\xi)\xi \ell/r^2$ indices.
  For small $\xi$, this is at most a fraction~$3\xi$ of the intersection, (which has at least $(1-\xi)\ell/r^2$ indices by the lemma).
  In fact, $\xi \le 1/8$ is sufficient for this.

  We conclude that each time~$R_u$ is blamed, we obtain from~$R_v$ an index set~$J$ that satisfies the 
  inequality of the second item. Hence, each region is blamed at most $O(r)$~times.
  The proposition is proven.
\end{proof}

\subsection{Allocating requests with blame}\label{ss:blamedRequests}

We extend Bob's strategy to also allocate the remaining requests and strengthen proposition~\ref{prop:strat_blaming}.

\begin{corollary}\label{cor:strat_blaming}
  Assume $d$ is small, $e$ is large and $\eps_i \le 1/(em^3n)$ for all $i \le m$.
  Bob has a strategy in \textnormal{variant E} in which each request is allocated.
\end{corollary}

\begin{proof}
  We use the previous strategy with the following modification. Each time a region is assigned, 
  we assign an  {\em extra} region. We refer to the first region as the  {\em normal} one.
  The normal region is used in a first attempt of an allocation, 
  and the associated extra region is used when the normal one is blamed for a failed allocation.
  We choose $\mcI$ in the same way as in the proof of  proposition~\ref{prop:strat_blaming}, 
  but with a twice larger value of~$r$, since we need twice as many regions.

  \medskip
  \noindent
  {\em Bob's strategy given a request $(\{u,v\},\eps)$.} 
  \begin{itemize}[leftmargin=1.3em,topsep=3pt,label=--,itemsep=0pt,parsep=3pt]
    \item 
      Apply the strategy from the previous section using the normal regions. If the allocation succeeds, we are finished.  
    \item  
      Otherwise, replace the normal region of the blamed string by its extra copy.  If the intersection has a free interval, allocated it, and finish the strategy. 
    \item 
      Otherwise, also replace the other region by its extra copy and allocate an interval in the intersection. (We will show that this is always possible.)
  \end{itemize}
  {\em End of the strategy.}

  \medskip
  \noindent
  By a similar argument as for  proposition~\ref{prop:strat_blaming} we conclude that for small $d$, 
  each normal region $R_u$ is blamed at most $O(r)$ times for a failed allocation in the first step. 

  We need to show that the strategy always succeeds in the second or third step.
  Assume that $R_u$ was blamed and was replaced by its extra copy.
  We prove that if an allocation attempt fails in the second or third step, then the extra region can not be ``blamed'',
  which means that the number of full blocks in an extra region is always less than half of the common blocks.

  Since the region $R_u$ is blamed at most $O(r)$ times, the extra region allocates at most the same number of requests.
  All allocated intervals have size at most $1/(2\ell)$ 
  and this is much smaller than $1/r^3$ for large~$n$. Thus, the total measure of allocated intervals in an extra 
  region is much smaller than $r/r^3 = 1/r^2$. On the other hand, the intersection of any 2 regions has size close to $1/r^2$. 
  Thus in an extra region, less than half of the common blocks are full and hence, it can not be blamed.
  
  This also implies that in the last step an allocation must happen, since neither extra region can be blamed. 
  The corollary is proven.
\end{proof}

\medskip
\noindent
Recall that in the game of section~\ref{ss:gameExact} we may assume that 
request sizes are at least $2^{-n-2}$ and are powers of 2, see remark~\ref{rem:makeGameFinite}. 
Using $m = n$, we almost obtain the strategy required for the condition of lemma~\ref{lem:strategyImpliesTheoremExact}. 
We obtain a strategy for some $\delta(n)$ proportional to~$n^{-4}$. 
It remains to improve this to~$n^{1.01}$.
For this we run 2 different copies of a strategy given by corollary~\ref{cor:strat_blaming}. 
The first one handles requests with sizes between $2^{-n-2}$ and $cn^{-4}$, and is obtained with $m = n$, 
The second handles request for sizes between $cn^{-4}$ and $c(4\log n)^3/n$, obtained with $m = 4\log n$.
This strategy satisfies the conditions of lemma~\ref{lem:strategyImpliesTheoremExact}, 
and theorem~\ref{th:exact} is proven (except for the combinatorial lemma).

\subsection{Proof of the combinatorial lemma}\label{ss:combinatorial}

We restate  lemma~\ref{lem:extractorlike}. 
\begin{samepage}
  \begin{lemma*}
  \combLemmaOne
\end{lemma*}
\end{samepage}

\begin{proof} 
  We use the probabilistic method and generate the list $\mcI = [I_1, \ldots, I_N]$ randomly as follows: 
  for each $(i,j) \in [\ell] \times [N]$, place $i$ into $I_j$ with probability~$1/r$.

  We first show that the first item of the lemma is not satisfied with probability less than~$1/2$.
  For this, we show that each of the 3 requirements is violated with probability less than $1/6$.
  We do this for the requirement on~$\#(I \cap J)$. For the other 2 requirements this is done similarly.
  For two fixed different indices $i,j \in [N]$, the expected value of $\#(I_i \cap I_j)$ is $\ell/r^2$.
  By the Chernoff bound, the probability that this deviates by more than a factor $1+\xi$, 
  is at most $2\exp(-\alpha \ell/r^2)$ for some $\alpha>0$. 
  By the union bound, the probability that the requirement is false, is at most 
  \[
    N^2 \cdot 2\exp(-\alpha \ell/r^2).
  \]
  This is less than~$1/6$ for large~$e$ in the assumption on~$\ell$.

  We prove a variant of the second item of the lemma. There exists a constant~$b$ such that 
  for each index set $I$ and for each $I' \subseteq I$ of size $(\tfrac 1 4 + 2\xi) \ell/r$, with 
  probability at most $1/2$ there exists a sublist $J_1, \ldots, J_{br}$ in $\mcI$ for which 
  \begin{equation}\tag{$*$}\label{eq:ggoal}
    \# \left(I' \cap J_j \right)  \;\ge\; (\tfrac 1 2 - 3\xi)\; \frac \ell {r^2}, \qquad j = 1,\ldots, br.
  \end{equation}
  This statement, together with the first item, implies the lemma.

  We first give the argument for larger sublists $J_1, \ldots, J_k$ of size $k = br^2$.
  Fix a set $I'$ of size $(\tfrac 1 4 + 2\xi)\ell/r$ and generate $k$ sets $J_1, \ldots, J_{k}$ randomly. Consider the event
  \[
    \left\{(i,j) : i \in I' \cap J_j  \right\}  \;\ge\; k \cdot (\tfrac 1 2 - 3\xi) \cdot \frac \ell {r^2},
  \]
  The expected value of the left-hand side is $k\# I'/r$, 
  and this is smaller than the right-hand side by a constant factor, (almost a factor $2$ for small $\xi$). 
  The Chernoff bound in multiplicative form implies that the above event happens with probability at most 
  $\exp(-\beta k\ell/r^2)$ for some~$\beta>0$.
  We need to consider the probability that this happens for 
  any choice of $I$ in $\mcI$, $I' \subseteq I$ and sublist $J_1, \ldots, J_{k}$ of~$\mcI$. 
  By the union bound, the probability that \eqref{eq:ggoal} is satisfied, is at most 
  \begin{equation*}
    N \cdot 2^{\ell} \cdot N^k \cdot \exp(-\beta k \ell/r^2),
  \end{equation*}
  where we used that the number of choices for $I' \subseteq I$ is at most $2^{\# I}$ with $\# I \le \ell$.
  To show that this probability is less than $1/2$, we show that
  \begin{align}
    (k+1)\log N &\le \tfrac 1 2 \frac{\beta k \ell}{r^2} \nonumber\\
    \ell &\le \tfrac 1 2 \frac{\beta k \ell}{r^2} \tag{$**$}\label{eq:change}.
  \end{align}
  The first item follows from $k+1 \le 2k$ and the assumption on~$\ell$. 
  The second item follows for $k = br^2$ and large~$b$.

  To obtain the argument for the smaller value $k = br$, 
  we use a better bound for the number of subsets $I' \subseteq I$. 
  We show that $\# I \le 2\ell/r$ with sufficiently small probability.
  More precisely, by the Chernoff bound, the logarithmic probability 
  that the size $\# I$ exceeds twice its expected value~$\ell/r$ 
  is proportional to $\ell/r$, which in turn is proportional to $k\ell/r^2$ if $k = br$.
  Thus, the probability of this event is much smaller than the probability in the union bound.
  Therefore, we may replace the left-hand side in~\eqref{eq:change} by $2\ell/r$, and this is indeed 
  satisfied for $k = br$ and large~$b$.
\end{proof}

\subsection{Variants of theorem~\ref{th:exact} for strings of different lengths}\label{ss:exactStrong}

We prove a more general version of  theorem~\ref{th:exact} that implies the following 2 results. 

\begin{corollary}\label{cor:exact_loglog}
  The prefix information distances are equal to $\E(x,y) + O(\log \log \KP(x,y))$. 
\end{corollary}

\noindent
For $r \le 0$, let $\log r = 0$. 

\begin{corollary}\label{cor:exact_length}
  The difference between each prefix information distance and $\E$ is at most linear in
  \[
    \E\left(\big\lceil \log |x| \big\rceil, \big\lceil\log |y|\big\rceil\right) \;+\;  \log \big(1.01\log |xy| - \E(x,y)\big). 
  \]
\end{corollary}

\noindent
It is not too difficult to prove these corollaries using theorem~\ref{th:exact}.
However, we will obtain them as special cases of an even more general result.
This result is formulated using jointly conditional complexity, 
which was first studied in~\cite{MuchnikCodes} and later in~\cite{KolmLogical,Vereshchagin2011joint}. 

\begin{definition}\label{def:jointlyConditional}
  The {\em jointly conditional complexity} of $z$ given a set $S$ on a machine $U$, is
  \[
    \KS_U(z \leftarrow S) \;=\; \min \left\{|p|: \forall y\in S [ U(p,y) \mathop{=} z ] \right\}.
  \]
\end{definition}

\noindent
We fix a prefix-free machine $U$ that minimizes the above function $\KS_U$ up to additive $O(1)$ constants and define $\KP(z \leftarrow S) = \KS_U(z \leftarrow S)$.
A similar definition could be given for prefix-stable machines, 
but we need to concatenate programs, (and on prefix-stable machines we do not know how to recover the splitting point).
In the above definition we can replace $z$ by a pair of integers.
The following result is equivalent to theorem~\ref{th:exact}.

\begin{theorem}\label{th:exact_strong}
  If $\ell \ge 1.01 \log \KP(x,y)$ then each prefix information distance exceeds $\E(x,y)$ by at most
  \[
    \KP\big(\ell,d_+ \,\leftarrow\, \{x,y\}\big) \;+\; O(1),
  \] 
  where $d_+ = \max\big(0, \ell \mathop{-} \E(x,y)\big)$.
\end{theorem}

\noindent
Theorem~\ref{th:exact}, corollary~\ref{cor:exact_loglog} and corollary~\ref{cor:exact_length} are special cases obtained
by setting $\ell = 1.01\log (3n)$, $\ell = \lceil 1.01\log \KP(x,y)\rceil$ and $\ell = \lceil 1.01 \log (3\max(|x|,|y|)) \rceil$.

\begin{proof}
  Let $\D$ be the largest distance, i.e., the prefix-free non-bipartite distance. We also use the conditional variant of this distance
  $\D(x,y \cnd \ell)$ given by the minimal length of a program $p$ that maps  $(x,\ell)$ to $y$ and $(y,\ell)$ to $x$. 
  Let $\E(x,y \cnd \ell) = \max (\KP(x \cnd y,\ell), \KP(y \cnd x,\ell))$.
  We consider 3 cases. 

  \medskip
  \noindent
  \textbf{Case $\ell \le \E(x,y \cnd \ell)$.} By concatenating programs we have
  \[
    \D(x,y) \;\le\; \KP(\ell \leftarrow \{x,y\}) \;+\; \D(x,y \cnd \ell)  \;+\; O(1).
  \]
  We show that $\D(x,y \cnd \ell) = \E(x,y \cnd \ell) + O(1)$,
  and this finishes the proof, since dropping $\ell$ from the condition can only increase~$\E$.

  Let $n = 2+2^{\lceil \ell/1.01\rceil}$. 
  Enumerate all pairs $(u,v)$ with $1.01\log \KP(u,v) \le \ell$ and associate all enumerated strings to strings of length~$n$.
  By choice of $n$, each enumerated string $u$ can be associated to a unique $n$-bit string, which we denote as~$u'$.
  We have $\D(x,y \cnd \ell) = \D(x',y' \cnd \ell) + O(1)$. By theorem~\ref{th:exact} conditional to $\ell$, (note that the proof indeed relativizes), 
  this is equal to $\E(x',y'\cnd \ell) + O(1)$ and hence, equal to $\E(x,y \cnd \ell) + O(1)$. 

  \medskip
  \noindent
  \textbf{Case $\ell \ge \E(x,y)$.} 
  The sum of $\ell$ and $d_+$ equals $\E(x,y)$. 
  The idea is to convert a program for a plain machine to a program for a prefix-free one by prepending 
  a prefix-free description of the length of the program.
  Let $e = \E(x,y)$. By concatenating programs we have 
  \[
    \D(x,y) \;\le\; \KP(e \leftarrow \{x,y\}) \;+\; \D(x,y \cnd e) \;+\; O(1).
  \]
  By the case assumption, the first term is bounded by the jointly conditional complexity term in the theorem.
  It remains to prove that $\D(x,y \cnd e) \le \E(x,y \cnd e) + O(1)$, since this is at most $\E(x,y) + O(1)$. 

  This follows from the plain characterization of the information distance given in theorem~\ref{th:bglvz}. 
  Indeed, we have $\KS(v \cnd w,e) \le \KP(v \cnd w,e) + O(1)$, thus $\max \big(\KS(x \cnd y,e), \KS(y \cnd x,e)\big) \le e + O(1)$.
  This implies that the plain distance is at most $e + O(1)$. Hence the prefix distance $\D(x,y \cnd e)$ is at most $e + O(1)$ as well.

  \medskip
  \noindent
  \textbf{Case $\E(x,y \cnd \ell) \le \ell \le \E(x,y)$.}
  In this case we concatenate a shortest program in the definition of $\KP(\ell \leftarrow \{x,y\})$ and a program 
  of length $\ell + O(1)$ for a plain machine that maps $(x, \ell)$ to~$y$ and $(y, \ell)$ to~$x$.
  The last program is obtained from theorem~\ref{th:bglvz} conditional to~$\ell$ using $\E(x,y \cnd \ell) \le \ell$.
  The total length is $\KP(\ell \leftarrow \{x,y\}) + \ell$ and by the assumption $\ell \le \E(x,y)$ this is sufficient for the theorem.
\end{proof}

\section{Non-shared information in a set of strings}\label{sec:nonsharedSet}

The characterization of the information distance is a non-trivial example where solving 2 different tasks 
simultaneously by a single program is not `harder' than solving the `hardest' of both tasks, 
where `hard' is understood in terms of program length.
We can generalize this to more tasks. 
In particular, we consider the problem of producing a set, given any of its elements. 

The information ``distance'' in a list of strings was defined in~\cite{lzlm} as a measure 
for the total information in the list that is not shared by all its items. 
We will use a simplified variant defined for a finite set $S$ of strings
\begin{equation}\tag{$*$}\label{eq:defDistSet}
  \D_U(S) \;=\; \min \left\{|p| : \forall x \in S [U(p,x) = S] \right\}.
\end{equation}
It is not clear to me why this can be called a ``distance'' unless we restrict to the case where $\#S = 2$, but we will remain consistent with the literature.
In this section we study the characterization of this measure as the maximum of conditional complexities.

\begin{remark}\label{rem:variantsSetDist}
  Four papers have studied this notion and they all use (slightly) different definitions.
  In~\cite{lzlm} the distance is defined for lists and prefix-free machines $U$ are used. 
  More precisely, given a list $x_1, \ldots, x_s$, they consider the minimal length of a program $p$ that for all $i,j \le s$
  satisfy $U(p,x_i,j) = x_j$. Thus the program $p$ might not be able to indicate the size of~$S$. 

  In~\cite{VitanyiMultiples} the distance was defined almost identically, but for multisets instead of lists.\footnote{
    In this paper the term ``list'' is redefined as a finite sequence of strings presented in lexicographic order. 
    Such a list contains the same information as the multiset defined by its items.
  }
  In~\cite[definition 3.13]{mahmud} the requirement for the program of $p$ is stronger than in~\cite{VitanyiMultiples} 
  because it must also provide the set size.\footnote{Also sorting requirements on the output differ, but they affect the measure by at most an additive constant.}

  Finally, in~\cite{Vitanyi2017} a variant with plain complexity is used. It differs from the above definition because this paper considers multisets and
  programs have the size $\# S$ as extra input. To make the definition easier, we have dropped this extra input and use ordinary sets.
  However, all results in section~\ref{sec:nonsharedSet} still hold if multisets are used. 
\end{remark}

\subsection{Plain variant}

We first consider the plain variant.
Thus, we fix some plain machine $U$ for which $\D_U$ in \eqref{eq:defDistSet} is minimal up to an additive constant, and drop the index.
Note that for sets of $2$ elements, the distance $\D(\{x,y\})$ is equal to the plain information distance from section~\ref{sec:plain}, (up to an additive constant). 

In~\cite{Vitanyi2017} a characterization with additive precision $\log \# S + O(1)$ was given for the plain variant of the distance discussed in remark~\ref{rem:variantsSetDist}.
It was also shown that this precision could not be decreased by more than an additive term. 
We present these 2 results with slightly worse precision.\footnote{ 
In our definition, programs do not have $\#S$ as input, and hence our distance is larger by an additive $O(\log \#S)$ term.
}
For later reference, we also present short formulations of their proofs.\footnote{
  In~\cite[Theorem 2]{lzlm} the above characterization was claimed to hold with precision logarithmic both 
  in the length of the strings and the size of~$S$. (The prefix variant was used, but this does not matter with the claimed level of precision.) 
  I was unable to understand their proof.
  Independently, in~\cite[Theorem 3.14]{mahmud}, the same characterization was given for the prefix variant with precision $O(\# S\log (\# S))$, 
  and this precision does not depend on the length of the strings in~$S$. 
  We already explained that even for $\#S = 2$ this is impossible, and the given proof contains a similar mistake.
}

\begin{theorem}[\cite{Vitanyi2017}]\label{th:setDistanceVit} 
  For all $S \subseteq \{0,1\}^*$
  \[
  \D(S) \;=\; \max \{\KS(S \cnd w) : w \in S \} \;+\; O(\log \# S).
\]
\end{theorem}

\begin{proof}
  The proof is similar to the proof of theorem~\ref{th:bglvz}. 
  The $\ge$-inequality follows by definition, 
  thus we only need to prove the $\le$-inequality.
  First we prove the inequality for sets $S$ of any fixed size~$s$. 
  Consider the following relation on sets of size~$s$:
  \[
    R_n(S) \quad \Longleftrightarrow \quad \forall x \mathop \in S \big[\KS(S \cnd x) < n\big]. 
  \]
  The collection of sets $S$ that satisfy this relation can be enumerated.

  Consider a hypergraph whose vertices are strings. Thus, there are infinitely many vertices.
  The hyperedges are given by sets $S$ of vertices for which $R_n(S)$ holds.
  Note that every vertex is incident on less than $2^n$ hyperedges.

  We assign colors to the hyperedges such that every two hyperedges that have at least 1 common node, have different colors.
  For this, we enumerate all hyperedges and color them in a greedy way using $s2^n$ colors.
  This number of colors is enough, because on each of its $s$ vertices  
  less than $2^n$ other hyperedges are incident. Thus less than $s2^n$ colors are unavailable, and at least one remains.

  Finally we construct a machine $V$ such that $\D_V(S)$ satisfies the inequality. 
  We first do this for a fixed value of~$s$.
  On input a program $p$ and a string $w$, the machine constructs the above hypergraph with $n = |p|-\lceil \log s\rceil$. 
  It interprets $w$ as a node and the program $p$ as a color. As soon as an edge $S$ is enumerated that is incident on $w$ 
  and has color $p$, it outputs the set $S$ and halts.
  By choice of $R_n$, the value of $\KS_V(S \cnd w)$ is at most the right hand side of the theorem. 

  If $s$ is not fixed, we prepend a prefix-free description of $s$ to the program~$p$.
  For example, we can use descriptions of length exactly $2\lceil\log s\rceil + 1$, and after this change, we can still recover $n$ 
  from the length of the program.
\end{proof}

\noindent
The following result implies that the $O(\log \# S)$ precision in the above theorem can not be decreased by more than a constant factor.
It is a variant of~\cite[Theorem 2.1]{Vitanyi2017} that follows by the same proof.

\begin{theorem}
  \label{th:precisionSetDistance}
  For all $s$, there exists a set $S \subseteq \{0,1\}^{\lceil 2\log s\rceil}$ with $s$ elements such that
    \[
      \max_{x \in S} \KS(S \cnd x) \;\le\; \log s + O(1) \;\le\; \D(S) - \log s + O(1). 
  \]
\end{theorem}

\begin{proof}
  A projective plane over a finite field consists of
  a set of points, a set of lines, and an incidence relation between lines and points.
  For a field of size $q$, we use the following properties of such a plane:
  \\(a) every line contains $q+1$ points, 
  \\(b) the number of points is equal to $q^2 + q + 1$,
  \\(c) every point lies on precisely $q+1$ lines, 
  \\(d) every 2 different lines intersect in precisely 1 point, 
  \\(e) the number of lines is equal to $q^2 + q + 1$.

  We will select a line and prove the claim for the set $S$ given by all points on this line. 
  By~(a), such an $S$ has size $s = q+1$. By~(b), its elements can be represented as strings 
  of length $\lceil 2\log s\rceil$. 
  If a point~$x$ lies on a line~$L$, then
  \[
    \KS(L\cnd x) \le \log s + O(1),
  \]
  by~(c).
  For every line~$L$, let~$p_L$ be a shortest program that on input any point from~$L$, produces~$L$.
  For different lines $L$, the programs $p_L$ must be different, because of~(d), 
  (for any 2 different lines $L$ and~$L'$, the intersection~$x$ satisfies $U(p_L,x) \not= U(p_{L'},x)$).
  Hence, by~(e),
  there must be a line~$L$ with $|p_L| \ge \lfloor\log q^2\rfloor = 2\log s - O(1)$.
  This implies $\D(L) \ge 2\log s - O(1)$ and hence the lemma.
\end{proof}

\begin{remark}
  For each constant $c$ and for $\#S \ge \Omega(\KP(S)^{1/c})$, the result of theorem~\ref{th:setDistanceVit} also follows from a Kolmogorov complexity variant of the 
  Slepian-Wolf theorem given in~\cite{BauwensZimandUniversal}. This variant states that there
  exists a randomized compression algorithm $\mathcal{C}$ such that: \footnote{
  Remarkably, there also exists such an algorithm that runs in time polynomial in $|z|/\eps$ and produces slightly longer strings of length $k + O(\log^2 (|z|/\eps))$.
  }
  \begin{itemize}[leftmargin=*]
    \item On input a string~$z$, a target size~$k$ and an error bound~$\eps$, 
      with probability 1 the algorithm produces a string $\mathcal{C}_{\eps,k}(z)$ of length $k + O(\log (|z|/\eps)$.
    \item 
      For some machine $M$ and for each string $w$ with $\KS(z \cnd w) \le k$, with probability $1-\eps$
      over the randomness in $\mathcal C$ we have 
  \[
    M(\mathcal{C}_{\eps,k}(z),w) \mathop{=} z.
  \]
  \end{itemize}
  To obtain the program that proves the bound of theorem~\ref{th:setDistanceVit}, 
  let $z$ be a shortest program of $S$, 
  $\eps = 1/(2\#S)$ and $k = \E(S)$.
  By the union bound, the probability that the program fails to print $z$ for some $w \in S$ is at most~$1/2$.
  Hence the required program of length $\E(S) + O(\log (n\# S))$ exists. 
  The given precision equals~$O(c\log \# S)$ by the assumption \mbox{$\#S \ge \Omega(\KP(S)^{1/c})$}.
\end{remark}

\subsection{Prefix variant}\label{ss:prefixSet}

The prefix-stable and prefix-free versions of the distance are obtained by fixing optimal machines of the respective types in~\eqref{eq:defDistSet}. 
Let 
\[
  \E(S) = \max_{x \in S} \KP(S \cnd x)\,.
\]
We restate theorem~\ref{th:exactSet} using these definitions.

\begin{theorem*}
  If $\E(S) \ge (\param\# S)^{\#S}\log n$ for some $S \subseteq \{0,1\}^n$,
  then both prefix set distances are equal to~$\E(S) + O(\log \# S)$.
\end{theorem*}

\noindent
{Remarks.}
\begin{simpleitem}
  \item 
    The $O(\log \# S)$ precision of the equality can not be decreased by more than a constant factor,
    since theorem~\ref{th:precisionSetDistance} also holds for the prefix distances by the same proof. 
    Note that the lower bound condition on $\E(S)$
    can be satisfied by appending a long enough random string to an element of~$S$.
  \item 
    We do not know whether the statement also holds for smaller values of $\E$, for example, if $\E(S) \ge c\log (n\#S)$ for some constant~$c$.
  \item 
    In proposition~\ref{prop:exactSet} the same equality is proven with a different condition
    for the set $S \subseteq \{0,1\}^n$: each pair $(u,v)$ of different strings 
    in $S$, should satisfy $\E(u,v) \ge \param\# S \log (n\#S)$. This condition is incomparible with the one of the theorem.
\end{simpleitem}

\medskip
\noindent
In the remainder of this paper we prove theorem~\ref{th:exactSet}. 
We first present a game. In the following subsections, we present a sequence of strategies for Bob 
that become increasingly stronger, but also require smaller request sizes.
From the strategy in section~\ref{ss:blamingFriends}, we obtain Proposition~\ref{prop:exactSet}.
All strategies rely on the same combinatorial lemma that is proven in the last subsection.

\bigskip
\noindent
{\em Description of the game.} The game for  theorem~\ref{th:exactSet} is similar as for theorem~\ref{th:exact}. 
Its parameters are: the set size~$s$, the length~$n$, $d>0$, and a list of request sizes $\eps_1, \ldots, \eps_m$ of length~$m$.
Alice's requests are pairs $(S,\eps)$ where $S$ is a set containing~$s$ strings of length~$n$, and~$\eps$ belongs to the list of request sizes.
Alice's requests must satisfy the following restriction:
for each string $u$, the sum of the sizes $\eps$ for all requests $(S,\eps)$ with $u \in S$, should be at most~$d$.
For all sets $S$, Bob maintains a subset $M_S$ of the Cantor space.
For each request $(S,\eps)$, he needs to enumerate an interval of size at least~$\eps$ into~$M_S$, unless such an interval was enumerated previously.
For each $S'$ that intersects $S$, he needs to enumerate an interval that is disjoint from~$M_{S'}$.
(As before, he must add an interval that does not overlap with the current set~$M_S$. 
Thus, this requirement must also hold for~$S = S'$.)

\begin{lemma}\label{lem:strat_implies_exactSet}
  Let $m = sn$ and $g(s) = (\param s)^s/2$. 
  Suppose there exists a polynomial $p$ such that Bob has a winning strategy in all games with $d = 1/p(s)$ 
  and $\eps_i \le (sn)^{-g(s)}$, then theorem~\ref{th:exactSet} is true.
\end{lemma}

\begin{proof}[Proof sketch.]
  The proof is very similar as for  lemma~\ref{lem:strategyImpliesTheoremExact}. We may assume that $n \ge s$, 
  because if $n \le s$, the result follows from the characterization for plain complexity.
  We use an approximation of $\E(S)$ to create a strategy for Alice, let it play against Bob's winning strategy, 
  and use Bob's moves to construct a machine that satisfies the conditions of the theorem.
  Note that $\E(S) \le (s-1)n + O(1)$, since each conditional complexity is bounded by this value. 
  Thus for large $n$, we may assume that $\E(S) \le sn$.

  Each time some value $\E(S)$ is updated to value $k$ with $(sn)^{-g(s)} \le d2^{-k} \le d2^{-sn}$, 
  Alice makes a request $(S,d2^{-k})$.
  This strategy is played against Bob's winning strategy, and these moves provide a distance that does not exceed~$\E(S) + O(\log \# S)$.
  We conclude that for all $S$ with $\E(S) \ge g(n) \log (snm)$, 
  the equality holds with precision $\log (1/d) + O(1) \le O(\log s)$, since $d$ is polynomial in~$1/s$.
  The theorem holds using $g(s)\log (sn) \le 2g(s) \log n$.
\end{proof}

\subsection{A strategy with leaders}\label{ss:blamingLeaders}

\begin{proposition}\label{prop:exactSetEasy}
  If $S \subseteq \{0,1\}^n$ such that for all pairs $(u,v)$ of different strings in $S$ we have $\KP(v \cnd u) \ge 3\#S\log (n\#S)$,
  then the prefix distances are equal to $\E(S) + O(\log \# S)$.
\end{proposition}

\noindent
To prove this proposition, we adapt the game.
Bob may decide not to allocate a request $(S, \eps)$, 
but if he does so, he must blame one string in~$S$, 
and declare at least one other string in $S$ to be the {\em leader} of the blamed string. 
During the game, a string receives more and more leaders, and 
Bob needs to satisfy the following restriction for some large constant~$c$ (independent of $n,m,s$):

\begin{center}
  {\em each string has at most $s^2(cm)^{s}$ different leaders.}
\end{center}

\noindent
In the usual way, the proposition follows from a winning strategy for Bob with $d$ polynomial in $1/s$, $m = sn$,
and $\eps_i \le (csn)^{s+O(1)}$ for large~$c$.
Indeed, if $u$ is a leader of $v$,
the above requirement implies that $\KP(u \cnd v) \le s\log (sn) + O(s)$. 
Thus no allocation is required for requests $(S,\eps)$ with sets $S$ that contain both~$u$ and~$v$.
Similarly, if $\eps \ge (csn)^{-s-O(1)}$ in a request $(S,\eps)$, 
then $\KP(u \cnd v) \le (s+O(1))\log sn + O(s)$ for all different $u$ and $v$ in~$S$, 
and again no allocation is required.

\bigskip
\noindent
The strategy is similar as for theorem~\ref{th:exact} and is based on collecting requests of the same size in contiguous areas. 
The Cantor space is partitioned into blocks of equal size and regions are assigned to them dynamically,
using a list $\mcI$ of index sets that satisfies some combinatorial properties.
Each time we need to assign a new region for a string, we select a fresh set of indices from the list, 
and define the region to be the unused blocks with these indices.
A block is  {\em full} if at least a fraction $1/s$ of its measure is allocated.

\bigskip
\noindent
\begin{samepage}
\noindent
{\em Bob's strategy to allocate a request~$(S,\eps)$.}
\begin{enumerate}[topsep=2pt,leftmargin=1.63em]
  \item \label{item:select}
    For each $v$ in $S$ select a region $R_v$ of $v$ with request size~$\eps$ for which the fraction of full blocks is less than~$1/(2s)$.
    If such a region does not exist, assign a new region for size~$\eps$.

  \item \label{item:block} 
    Find a free interval that intersects all selected regions~$R_v$.
    If such an interval is found, allocate it and terminate the strategy. 
    Otherwise, select a region $R_u$ for which at least a fraction $1/s$ of the common blocks are full.
    Blame this region~$R_u$.

  \item  \label{item:blame} 
    If $S\setminus \{u\}$ contains a leader of~$u$, then nothing needs to be done and the strategy is terminated.
    Otherwise, declare all elements of $S \setminus \{u\}$ to be leaders of~$u$. 
\end{enumerate}
\noindent
{\em End of the strategy.} 
\end{samepage}

\medskip
\noindent
%
We now present the combinatorial properties that the list~$\mcI$ should satisfy.
A $t$-selection is a subset of $[N]$ of size~$t$.
Given a nonempty subset $T \subseteq [N]$, let $\mcI[T] = \bigcap_{j \in T} \mcI_j$.

\newcommand{\combLemmaSet}{
  Let $\xi$ be small and $e$ large. For all $s \ge 2, r,N$ and $\ell \ge es r^{s+1}\log N$, 
  there exists a list $\mcI$ of $N$ subsets of $[\ell]$ such that
  \begin{itemize}[leftmargin=*]
    \item 
      For all $t \le s+1$ and each $t$-selection $T$: \[1-\xi \;\le\; \frac{\#\mcI[T]}{\ell/r^t}  \;\le\; 1+\xi.\]

    \item For every $I$ in $\mcI$, every $I' \subseteq I$ of size $\#I/(2s)$, there exist at most $k = O(sr^{s-1})$ pairwise disjoint $(s{-}1)$-selections $T$ such that
      \[
	\# \Big(I' \cap \mcI[T]\Big) \;\ge\; \tfrac{1-\xi}{s} \;\# \Big(I \cap \mcI[T] \Big).
	\]
  \end{itemize}
}
\begin{lemma}\label{lem:combLemmaSet}
  \combLemmaSet
\end{lemma}

\begin{proof}[Proof of  proposition~\ref{prop:exactSetEasy}.] 
  Note that it is enough to show the lemma for $s$ being a power of~$2$ (this allows us to avoid explicit rounding of interval sizes).
  We apply the lemma with $r=\tilde{c}m$ for some large~$\tilde{c}$ that we determine shortly, with $N = r2^n$, and with $\ell$ being the smallest power of 2 that exceeds the lower bound of the lemma.
  We show that the above strategy satisfies the requirement of the adapted game 
  for all $n$, $s \ge 2$, $d$ proportional to $1/s^4$ and request sizes $\eps \le 1/(\ell s)$, i.e., $\eps \le O(1/(s^2n (cm)^{s+1})$. 

  In a similar way as for proposition~\ref{prop:strat_blaming}, one can show that for large $\tilde{c}$ and small $d = \Theta(1/s^4)$,
  at most $\xi r/(2s^2)$ regions are assigned for each string.
  Thus, the list~$\mcI$ contains enough index sets for all assignments in the strategy.

  It remains to show that the requirement on the number of leaders is satisfied.
  Fix a region~$R_u$. 
  Each time~$R_u$ is selected and the number of leaders in step~3 increases,
  the selected regions of the strings in $S \setminus \{u\}$ define an $(s{-}1)$-selection $T$ of~$\mcI$ 
  given by the index sets used to assign these regions.
  By a similar analysis as for proposition~\ref{prop:strat_blaming}, 
  one can show that at most a fraction $s \cdot (1+\xi) \cdot (\xi/(2s^2))$ of indices are removed from the intersection during an assignment. 
  (Here we use the property that intersections of size $t = s+1$ contain approximately a fraction~$r^{-t}$ of the blocks.)
  We conclude that the $(s{-}1)$-selection satisfies the inequality of the second item of the combinatorial lemma.
  By construction of step 3, subsequent extensions correspond to disjoint selections.
  This implies that for the fixed region $R_u$, at most $O(sr^{s-1})$ times a tuple of $(s{-}1)$ leaders are declared.

  Since at most $r$ regions are assigned for~$u$, 
  the number of leaders is at most~$O(r \cdot s \cdot sr^{s-1})$, which is bounded by~$s^2 (cm)^s$ for large~$c$.
\end{proof}

\subsection{A strategy with pairs of friends}\label{ss:blamingFriends}

\begin{proposition}\label{prop:exactSet}
  For $S \subseteq \{0,1\}^n$ such that $\E(u,v) \ge 3\#S\log (n\#S)$ for all different $u$ and $v$ in~$S$, 
  the prefix distances are equal to $\E(S) + O(\log \# S)$.
\end{proposition}

\noindent
Again the game is the same as in section~\ref{ss:prefixSet} with a different requirement for Bob.
He may decide not to allocate a request $(S, \eps)$, 
but if he does so, he must declare two different strings in~$S$ to be  {\em friends} of each other. 
(One might think that these two strings are blamed together and this situation creates a friendship.
Moreover, in the strategy below, strings are declared friends when they appear together in too many ``difficult'' requests.) 
Being friends is a symmetrical relation on strings.
During the game a string collects friends. 
Bob's requirement limits the number of friends: for some constant~$c$ and $r = cm$

\begin{center}
  {\em each string can have at most $r^{2s}$ different friends.}
\end{center}

\noindent
Again proposition~\ref{prop:exactSet} follows by showing that there exists a winning strategy for Bob for some constant $c$, $d = \text{poly}(s^{-1})$, $m = sn$, 
and $\eps_i \le 1/(s^{2}n(cm)^{2s})$.
Indeed, the above requirement implies that if $u$ and $v$ are friends, then $\KP(u\cnd v) \le 2s\log (sn) + O(s)$ and similarly for  $\KP(v \cnd u)$; 
thus also $\E(u,v)$ satisfies this bound. Also, if $\E(S)$ is small, then $\E(u,v)$ is small for all $u$ and $v$ in~$S$.

\bigskip
\noindent
In the previous subsection, each string has few leaders. But some string might be the leader of many strings, in other words, he may have many followers.
Thus we can not allow all followers to be friends. Instead we use the following rule to create friendships for some fraction $f>0$ that we choose later:
2 strings become friends if the total measure of requests in which they appear together is at least~$f$. 
This implies that a string $u$ can have at most $(s-1) \cdot (d/f)$ friends, 
(because each request containing $u$ can increase the fraction of at most $s-1$ strings in a request).

Suppose we use the same strategy as before. It may now happen that a string is blamed for an unallocated request that contains no friends. 
Still, the above game requires us to allocate such requests.
For this, we make a modification similar as in section~\ref{ss:blamedRequests}.
Each time the strategy assigns a region, 
we also assign an associated copy that we call {\em extra region}. The original one is called {\em normal region}. 
If a request contains no pair of friends and can not be allocated in the selected regions, 
then one of the regions is blamed, 
and we replace the blamed region by its associated extra copy, and repeat the strategy.
We will ensure that only normal regions can be blamed. Hence, after at most $s$ repetitions, the strategy makes an allocation.
See below for the detailed strategy.

In order for this strategy to work, we must ensure that extra regions can never be blamed. 
We show that for small $f$,
the total measure of allocated requests in an extra region is at most~$r^s/(2s)$. 
This is $s$ times smaller than any intersection of $s$ regions, (by the combinatorial lemma with essentially the same parameters, see below), and hence, an extra region can never be blamed.
Indeed, each time a string $u$ is blamed, the request contains a leader. For each leader of $u$, 
the measure of requests containing this leader is at most $f$, (because the extra region never allocates requests containing friends).
Since a string can have at most $s^2r^s$ leaders, the total measure is at most $s^2r^s \cdot f$. 
Hence, it suffices to choose $f = r^{-2s}/(2s^3)$ to satisfy the requirement.
For $d = s^{-4}/2$, the number of friends a string can have is at most $sd/f = r^{-2s}$.

We choose all other parameters in the combinatorial lemma in the same way as in the previous paragraph, 
except $r$ is chosen twice larger, because we need a double amount of regions. 
For convenience, we present the full strategy.

\bigskip
\noindent
\begin{samepage}
  \noindent
\!\!{\em Bob's strategy to allocate a request~$(S,\eps)$.}
\begin{enumerate}[topsep=2pt,leftmargin=1.63em]
  \item \label{item:friends}
    Declare all pairs of strings in $S$ that coappear in at least a measure $f = r^{-2s}/(2s^3)$ of requests, to be friends.
    If $S$ contains 2 strings that are a pair of friends, terminate the strategy, (since no allocation is needed).

  \item \label{item:select}
    For each $v$ in~$S$ select a normal region~$R_v$ of~$v$ with request size~$\eps$ for which the fraction of full blocks is less than~$1/(2s)$.
    If no such region exists, assign a new normal and extra region for size~$\eps$.

  \item \label{item:block} 
    Find an interval of size $\eps$ that belongs to all regions $R_v$ and is free for all strings.
    If such an interval is found, allocate it and terminate the strategy. 
    Otherwise, select a string $u$ for which at least a fraction~$1/s$ of the common blocks is full, and
    {\em blame} the region~$R_u$. Replace region $R_u$ by its extra copy and repeat this step.

\end{enumerate}
\noindent
{\em End of the strategy.} 
\end{samepage}

\subsection{A strategy with groups of friends}\label{ss:blamingGroups}

We now present the strategy that implies theorem~\ref{th:exactSet}. 
It consists of $s$ substrategies that exchange a more general type of request.
Such a request is given by a pair $(C,\eps)$ 
where $C$ is a partition of some $s$-element set.
The sets in $C$ should always be nonempty and we refer to them as groups (of friends).
A request $(S,\eps)$ of the game is viewed as the partition~$C$ containing the $s$ singleton subsets of~$S$. 

The strategy in the previous section either allocates a request, 
or declares two strings to be friends. This last operation, we view as merging of 2 singleton groups.
Hence, we obtain a new request $(C',\eps)$ where $C'$ contains a group of 2 friends and $s-2$ singletons. 
More generally, given a request $(C,\eps)$, the idea is to run the strategy of the previous section 
and use a separate region for each group that appears. More precisely, if 3 strings form a group (of friends), 
we allocate the same region for the 3 strings, and associate this group to the region. 
The result is that either an interval is allocated or 2 groups are merged.
After at most $s-1$ iterations we either allocated the request or obtain a single group.
Such groups are allocated using separate blocks.
We show that each string can only belong to a few different groups. This allows us to use blocks of reasonably large size.
To implement this strategy, we need a way to assign regions in an online way, because we do not know in advance 
which groups will appear. Fortunately, such allocations we already obtained in the previous subsections.

\bigskip
\noindent
Now the details. 
We partition the Cantor space in $s$ approximately equal parts and run~$s$ substrategies in parallel.
Substrategy~$t$ receives requests $(C,\eps)$ where $C$ contains $t$ groups.
(Recall that the groups are nonempty and that $C$ is a partition of some $s$-element set.)
Also recall that a request $(\{x_1,\ldots, x_s\},\eps)$ of the game is given to substrategy~$s$ as 
$
  \big(\big\{ \{x_1\}, \ldots, \{x_s\}\big\},\eps\big). 
  $
If the $t$-th substrategy does not allocate an interval, then it produces a request for substrategy~$t-1$. 
Finally, substrategy~$1$ allocates each request using a separate block.

In the previous subsection, the strategy associates regions to request sizes. 
Now we associate regions to pairs $(\eps_i, S')$ of request size $\eps_i$ and group~$S'$. 
We refer to such pairs as  {\em labels}. 
Substrategy~$s$ (which is executed first) receives requests with singleton groups, 
thus for a string $u$, all labels are of the form~$(\eps,\{u\})$.
Therefore, the substrategy executes precisely the strategy of the previous subsection. 
The other substrategies operate similarly. 
If $S' \in C$, then we say that the label $(\eps, S')$ {\em appears} in the request $(C,\eps)$.

Let $m_{s} = sn$, and let $c$ be the constant from the previous subsection.
For all $t =  s, s-1,\ldots,2$, let $r_{t} = cm_{t}$ and $m_{t-1} = m_t \cdot (r_{t})^{2t}$.
We obtain the following.

\begin{quote} {\em
  For each string~$u$, there are at most $m_t$ different labels that appear in a request given to substrategy~$t$.
  }
\end{quote}

\noindent
This property is trivially true for $t = s$, since this stage only receives requests with singleton groups, and we only use 
request $sn$ different request sizes.
%
We say that {\em a block is full for a group} $S'$ if at least a fraction $\#S'/s$ of its measure is allocated.
The previous strategy is only changed in the first step, where a call is made to another substrategy. 
The other steps are almost identical.

\medskip
\noindent
\begin{samepage}
 {\em The $t$-th substrategy to allocate a request $(C,\eps)$.} 
 \begin{enumerate}[leftmargin=1.3em,topsep=3pt]
   \item 
     Declare every pair of groups in $C$ that coappear together in a measure $(r_t)^{2t}/2t^3$ of requests, to be friends.
     If $C$ contains a pair of friends, then generate a request $C'$ for strategy $t-1$ where the pair is merged to a single group.

   \item 
     For each group $S'$ in $C$, select a normal region $R_{S'}$ with label $(\eps,S')$ 
     for which less than a fraction $1/(2t)$ of blocks are full.
     If no such region exists, we assign a new normal and extra region with label $(\eps,S')$, and select the normal region.

    \item 
       If there exists a free common interval of size~$\eps$, allocate this interval and terminate the strategy.
      Otherwise {\em blame} the region $R_U$ for which the intersection contains at least a fraction $1/t$ of blocks that are full for~$U$.
      Replace the region $R_U$ by its associated extra region and repeat this step.
 \end{enumerate}
 {\em End of the substrategy.}
\end{samepage}
 
\bigskip
\noindent
We show that the requirement on the number of labels is satisfied.
We use downward induction on~$t$. For $t = s$ this is already proven. 
Assume that each string receives at most $m_t$ different labels in the substrategy $t$.  
By a similar analysis as before, it assigns at most $r_t = cm_t$ different regions for each string (in fact a factor $\xi/(2t^2)$ less). 
Thus the string becomes friends with at most $(r_t)^{2t}$ other groups. 
After merging, the string belongs to at most $m_t \cdot (r_t)^{2t}$ different groups.
The induction step is proven. 

In Substage 1, we obtain $r_1$ different requests, and for each request we divide the space of substrategy 1 in $1/(sr_1)$ blocks. 
The selection of such a block is easy. Hence, all requests with sizes $\eps \le 1/(s^2r_1)$ are allocated.
We can bound $s^2 r_1$ by 
\[
   (csn)^{(2s +1)^s}.
\]
We may assume $n \ge s$, 
and for large $n$ and $s \ge 2$ this is at most $(sn)^{(\param s)^s}$.
Hence, a winning strategy exists if all requests sizes are bounded by the inverse of this quantity.
To finish the proof of  theorem~\ref{th:exactSet} it only remains to prove the combinatorial lemma.

\subsection{Proof of the combinatorial lemma}\label{ss:combLemmasSet}

We restate lemma~\ref{lem:combLemmaSet}. 

\begin{lemma*}
  \combLemmaSet
\end{lemma*}

\begin{proof}
  The proof is similar as for lemma~\ref{lem:extractorlike}, and is repeated for convenience. We use the probabilistic method. 
  We assign $\mcI = [I_1, \ldots, I_N]$ as follows. For all $i \in [\ell]$ and $j \in [N]$, place $i$ in $I_j$ with probability~$1/r$.
  By the Chernoff bound in multiplicative form, the first item holds with probability more than~$1/2$. (The details are similar as for lemma~\ref{lem:extractorlike}.) 

  Let $b$ be a large constant that we determine later and let $k = bsr^{s}$. 
  We first prove the lemma with this weaker bound. Afterwards, we explain how we can decrease $k$ by a factor~$r$.
  Consider the following variant of the second requirement:
  for each index set $I \in \mcI$, each $I' \subseteq I$ of size $(1+\xi)\ell/(2sr)$, 
  there exists a list of pairwise disjoint $(s-1)$-selections $T_1, \ldots, T_{k}$ such that for all $j \in [k]$:
  \[
    \# (I' \cap \bigcap T_j) \; \ge\; \frac {(1-\xi)^2}{s} \frac{\ell}{r^s}.
  \]
  Together with the first requirement, this implies the second requirement of the lemma.
  After summing over $j \le k$, this statement implies 
  \[
    \# \big\{(i,j) : i \in I' \cap \bigcap T_j\big\}  \;\;\ge\;\; \frac{(1-\xi)^2}{s} \cdot \frac{k\ell}{r^s}.
  \]
  It suffices to show that this inequality holds with probability less than~$1/2$.
  The expected value of the left-hand side is at most $k\# I'/r^{s-1}$. 
  For small $\xi$, this exceeds the right-hand side by a constant fraction larger than~1.
  By the Chernoff bound in multiplicative form and the union bound, the probability that this inequality holds is at most
  \[
    2^{\ell} N^{sk} \exp(-\alpha \frac {k\ell} {sr^{s}})
  \]
  for some small constant $\alpha > 0$. 
  This is strictly smaller than~$1/2$ if
  \begin{align}
    \ell &\;\le\; \tfrac 1 2 \alpha \frac {k\ell}{sr^s} \tag{$*$}\label{eq:setBound}\\
    sk\log N &\;\le\; \tfrac 1 2 \alpha \frac {k\ell}{sr^s}.\nonumber 
  \end{align}
  By assumption on $\ell$ and the choice of $k$, these inequalities hold.

  To decrease $k$ by a linear factor, we bound the number of subsets $I' \subseteq I$ as $2^{2\ell/r}$. 
  We can do this, because the probability that $\#I > 2\ell/r$ can be neglected.
  Indeed, the probability of this event is $\exp( -\alpha' \ell/r)$ and for $k = bsr^{s-1}$ and large $b$, 
  this is at least proportional to the exponent in the union bound.
  With this better bound for the number of $I'$,  
  the left-hand side of  \eqref{eq:setBound} decreases by a factor~$r$, and for the given value of $k$, 
  this bound is satisfied as well.
\end{proof}

\section{Open questions}\label{sec:openQuestions}

In  section~\ref{sec:prefix-distance} we defined 4 prefix information distances on strings. 
We observed 4 trivial relations: the bipartite distances are bounded by the non-bipartite ones, 
and the prefix-stable distances are bounded by the prefix-free ones, (up to additive $O(1)$ constants). 
Under the assumptions of theorem~\ref{th:exact} they are all equal.

\begin{question}
  Which of these distances are always equal up to additive constants?
\end{question}

\noindent
Under the assumptions of theorem~\ref{th:exact} they are also equal to $\max (\KP(x \cnd y), \KP(y \cnd x)) + O(1)$ and hence, satisfy the triangle inequality.

\begin{question}
  Which of the 4 distances always satisfies the triangle inequality?
\end{question}

\noindent
The assumption of theorem~\ref{th:exactSet} requires the maximum to be at least $(\param s)^s \log n$, where $s = \# S$. 
We do not know whether this difference can have a double exponential improvement in~$s$.

\begin{question}
  Does  theorem~\ref{th:exactSet} hold under the weaker assumption that the maximum is at least $1.01 \log (ns)$? 
\end{question}

\bibliography{bib}
\end{document}